\documentclass[preprint,12pt]{elsarticle}
\usepackage[utf8]{inputenc}
\usepackage{amsthm}
\usepackage{amsmath}
\usepackage{amsfonts}
\usepackage{amssymb}
\usepackage{graphicx}
\usepackage{epstopdf}
\usepackage{color}
\usepackage{multirow}
\usepackage{mathtools}
\usepackage{xcolor}
\usepackage{a4wide}
\usepackage{bm}
\usepackage{caption}
\usepackage{subcaption}
\usepackage[
  separate-uncertainty = true,
  multi-part-units = repeat
]{siunitx}
\usepackage{rotating}
\usepackage{verbatim}
\usepackage{lineno,hyperref}
\modulolinenumbers[5]
\usepackage{tikz}
\tikzstyle{noeud} = [circle, draw, fill=white, inner sep=2pt]
\journal{Theoretical Computer Science}
\newtheorem{theorem}{Theorem}[section]
\newtheorem{lemma}{Lemma}[section]
\newtheorem{claim}{Claim}[section]
\newdefinition{remark}{Remark}[section]
\newtheorem{observation}{Observation}[section]
\newdefinition{definition}{Definition} \newdefinition{example}{Example}

\usepackage[noend,linesnumbered,ruled,lined]{algorithm2e}
\SetAlFnt{\scriptsize}
\usepackage{amsfonts}
\usepackage{makecell}

\usepackage[usenames,dvipsnames]{pstricks}
\usepackage{pstricks-add}
\usepackage{epsfig}
\usepackage{pst-grad} % For gradients
\usepackage{pst-plot} % For axes
\usepackage[space]{grffile} % For spaces in paths
\usepackage{etoolbox} % For spaces in paths
\makeatletter % For spaces in paths
\patchcmd\Gread@eps{\@inputcheck#1 }{\@inputcheck"#1"\relax}{}{}
\makeatother
\usepackage{mathtools}

\usepackage{graphicx}
\usepackage{tabularx}
\usepackage{textcomp}
\usepackage{xcolor}
\usepackage[labelformat=simple]{subcaption}

\usepackage{booktabs}

\usepackage{pgf,tikz,pgfplots}
\pgfplotsset{compat=1.15}
\usepackage{mathrsfs}
\usepackage{hyperref}
\usepackage{multicol}
\SetCommentSty{mycommfont}

\journal{NA}

\usepackage{a4wide}
\begin{document}
\usetikzlibrary{arrows}
\begin{frontmatter}

\title{Balanced Dispersion on Time-Varying Dynamic Graphs\footnote{A preliminary version of this work \cite{Saxena_2025} is accepted in ICDCN 2025.}}

\author[1]{Ashish Saxena}
\author[1]{Tanvir Kaur}
\author[1]{Kaushik Mondal\footnote{Corresponding Author}}
\affiliation[1]{organization={Department of Mathematics},%Department and Organization
            addressline={Indian Institute of Technology Ropar}, 
            city={Rupnagar},
            postcode={140001}, 
            state={Punjab},
            country={India}}

%% Abstract
\begin{abstract}
We aim to connect two problems, namely, dispersion and load balancing. Both problems have already been studied over static as well as dynamic graphs. Though dispersion and load balancing share some common features, the tools used in solving load balancing differ significantly from those used in solving dispersion. One of the reasons is that the load balancing problem is introduced and studied heavily over graphs where nodes are the processors and work under the message passing model, whereas dispersion is a task for mobile agents to achieve on graphs. To bring the (load) balancing aspect in the dispersion problem, we say, mobile agents move to balance themselves as equally as possible across the nodes of the graph, instead of stationary nodes sharing loads in the load balancing problem. We call it the \emph{$k$-balanced dispersion} problem and study it on dynamic graphs. This is equivalent to the load balancing problem considering movable loads in form of the agents. \\
\\
Earlier, on static graphs, the \emph{$k$-dispersion} problem [TAMC 2019] aimed for the same by putting an upper bound on the number of agents on each node in the final configuration; however, the absence of a lower bound on the number of agents in their problem definition hampers the load-balancing aspect, as some nodes may end up with no agents in the final configuration. We take care of this part in our \emph{$k$-balanced dispersion} problem definition and thus produce a stronger connection between the two domains.\\
\\
We study the \emph{$k$-balanced dispersion} problem on temporally connected dynamic graphs [TCS 2016] and $\ell$-bounded 1-interval connected dynamic graphs [JCSS 2021]. We provide several necessary and sufficient conditions to solve \emph{$k$-balanced dispersion} problem including a couple of impossibility results. 
\end{abstract}


\end{frontmatter}

\section{Introduction}\label{sec:intro}
The \textit{dispersion problem} was first introduced by Augustine et al.\ in \cite{Augustine_2018} for static graphs. The goal is to coordinate \( k \leq n \) mobile agents on a graph with \( n \) nodes so that no two agents occupy the same node. In dispersion, the objective is to minimize the total cost of reaching such a configuration, assuming that movement is relatively cheap compared to the cost of multiple agents sharing the same node. This problem has been explored under a wide range of assumptions: different agent interaction mechanisms like whiteboards, tokens, face-to-face meetings, or global communication; different timing models such as asynchronous, semi-synchronous, or fully synchronous rounds; and varying degrees of knowledge about the graph and memory limitations. For a broader perspective, see \cite{Augustine_2018, Ajay_2019, Molla_2019, ShintakuSKM20, KshemkalyaniS21, sudo, Molla_19, DAS_2023, Gorain_2024, Italiano_2025, Ajay_2025, Ajay_2022, Ajay_2025_, Patnayak_2025}. Dispersion is closely related to several classical problems. Scattering \cite{Elor_2011,SHIBATA_18,SHIBATA_22} on graphs requires that \( k \leq n \) agents spread out with equal spacing on symmetric structures such as rings or grids. This can be seen as dispersion with an additional constraint of equidistance. In multi-agent exploration \cite{Pelc_2004}, the goal is for agents starting from a common node to collectively visit all nodes in the graph as quickly as possible, which is a task that any dispersion algorithm with $k=n$ agents inherently solves. On the other hand, load balancing \cite{Cybenko_1989,Monien_2004} deals with nodes exchanging load to reach an even distribution. Given a graph $G$ representing the network with $n$ nodes, where
each node contains the load $w_i$, the goal is to move the load across the edges so that finally the weight of each node is (approximately) equal to 
\[
\overline{w_i} = \frac{1}{n} \sum_{j=1}^{n} w_j.
\]
\begin{figure}
    \centering
    \includegraphics[width=0.2\linewidth]{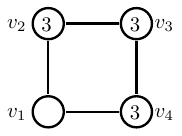}
    \caption{Configuration as per $k$-dispersion.}
    \label{fig:disp_1}
\end{figure}
Dispersion flips this perspective: instead of stationary nodes sharing load, mobile agents move to balance themselves across the network. While the tools used in load balancing differ significantly from those in mobile agent problems, studying dispersion can offer a common ground to connect these approaches and possibly transfer insights between the two domains. In 2019, Molla et al.\cite{Molla_19} established a connection between the problems of dispersion and load balancing by introducing the notion of \emph{$k$-dispersion}. Given a graph $G$ with $n$ nodes and $k \geq 1$ agents placed arbitrarily on its nodes, a configuration is said to be a $k$-dispersion if no node have more than $\lceil k/n \rceil$ agents. However, $k$-dispersion does not capture the idea of load balancing well enough. For example, consider \( n = 4 \), \( k = 9 \). Fig.~\ref{fig:disp_1} shows a valid configuration for \( k \)-dispersion as $\lceil k/n \rceil=3$. However, there is a node \( v_1 \) which has no agent. If parameters $k$ and $n$ are significantly large, there can be multiple nodes which has no agent. Therefore, agent load is not balanced. To build a stronger bridge between dispersion and load balancing, we introduce a constraint requiring each node to have at least \( \lfloor k/n \rfloor \) and at most \( \lceil k/n \rceil \) agents. Formally,

\begin{definition}
    \textbf{($k$-balanced dispersion)}
    Let $k \geq 1$ agents be placed on a graph $G$ with $n$ nodes. A configuration is said to be a balanced dispersion if each node has at least $\lfloor k/n \rfloor$ and at most $\lceil k/n \rceil$ agents.
\end{definition}

In modern networks, such as mobile ad-hoc networks, network changes are frequent. These systems are often referred to as dynamic graphs. To analyze such systems, many dynamic graph models have been proposed, each under different assumptions. In the synchronous setting, where time is considered in discrete rounds, a dynamic graph \( \mathcal{G} \) can be seen as a sequence of static graphs \( \mathcal{G}_0, \mathcal{G}_1, \mathcal{G}_2, \ldots \), where each \( \mathcal{G}_r \) represents a graph at round \( r \). This sequence is often referred to as a \textit{time evolving graph} or \textit{dynamic graph}. But if the graph evolves arbitrarily without any constraint, most distributed problems become unsolvable. To overcome this, several works introduce structural assumptions such as temporal connectivity~\cite{MICHAIL_2016}, $T$-interval connectivity~\cite{Kuhn_2010}, $T$-path connectivity~\cite{Saxena_2025_}, or connectivity time~\cite{Michail_2014}, which have enabled progress on fundamental problems like exploration, gathering, gossiping, and leader election.

A graph \( G \) is called port-labeled when each edge incident to a node \( u \) is assigned a locally unique port number from the set \( \{0, 1, \ldots, \deg(u) - 1\} \), where \( \deg(u) \) is the degree of node \( u \). That is, every undirected edge \( \overline{uv} \) has two independent port labels: one at node \( u \) and one at node \( v \). This port numbering is crucial for enabling agents to distinguish between their outgoing edges. Without port labels, especially in anonymous graphs where nodes have no IDs, all edges at a node would be indistinguishable. As shown in \cite{Pelc2012}, this makes even basic problems like exploration, dispersion, or rendezvous impossible to solve deterministically.

The question of how port numbers behave across rounds becomes even more critical in dynamic graphs. In the literature on dynamic graphs involving mobile agents, mainly two port-labeling models have been studied. The first, we call as \textit{fixed port labeling} \cite{gotoh2018group, gotoh2019exploration, GOTOH_2021, di2020distributed, bournat2017computability, Luna_2025}, the network has a static footprint graph \( G \), and each \( \mathcal{G}_r \) is a subgraph of \( G \). The outgoing ports at node $u$ in $\mathcal{G}_r$ are the same as the outgoing ports of node $u$ in $G$. Specifically, if port \( \lambda \) at node \( u \) leads to node \( v \) in the footprint $G$, then whenever that edge corresponding port appears in any \( \mathcal{G}_r \), it always leads to node \( v \). This provides local stability in the network. The second, we call \textit{degree-based dynamic port labelling} \cite{Ajay_dynamicdisp, Saxena_2025_}, where there is no footprint. Therefore, the degree of each node can change arbitrarily in each round, and port numbers are reassigned dynamically based on the current degree. As a result, the port $\lambda$ at node $u$ may lead to different nodes at different rounds.

In the next section, we formally describe the model, and on this model, we study \( k \)-balanced dispersion.

\section{Model}

\noindent \textbf{Dynamic graph}: A dynamic network is modeled as a \emph{time-varying graph (TVG)}, denoted by \( \mathcal{G} = (V, E, T, \rho) \), where \( V \) is a set of nodes, \( E \) is a set of edges, \( T \) is the temporal domain, and \( \rho : E \times T \rightarrow \{0, 1\} \) is the presence function, which indicates whether a given edge is available at a given time. The static graph \( G = (V, E) \) is referred to as the underlying graph (or footprint) of the TVG \( \mathcal{G} \), where \( |V| = n \) and \( |E| = m \). For a node \( v \in V \), let \( E(v) \subseteq E \) denote the set of edges incident on \( v \) in the footprint. The degree of node \( v \) is defined as \( \deg(v) = |E(v)| \), and the diameter of $\mathcal{G}_r$ is denoted by $D_r$, and $\widetilde{D}=\max_{r\geq 1} D_r$ denotes the dynamic diameter of $\mathcal{G}$. Nodes in $V$ are anonymous: they lack unique identifiers and possess no storage. Each edge incident to a node \( v \) is locally labeled with a port number. This labeling is defined by a bijective function \( \lambda_v : E(v) \rightarrow \{0, \ldots, \delta_v - 1\} \), which assigns a distinct label to each incident edge of \( v \). No further assumptions are made about the labeling. Assuming time is discrete, the TVG \( \mathcal{G} \) can be viewed as a sequence of static graphs \( \mathcal{S}_{\mathcal{G}} = \mathcal{G}_0, \mathcal{G}_1, \ldots, \mathcal{G}_r, \ldots \), where each \( \mathcal{G}_r = (V, E_r) \) denotes the snapshot of \( \mathcal{G} \) at round \( r \), with \( E_r = \{ e \in E \mid \rho(e, r) = 1 \} \). The set of edges not present at time \( r \) is denoted by \( \overline{E}_r = E \setminus E_r \subseteq E \). 

% the minimum degree of \( G \) is given by \( \delta = \min_{v \in V} \deg(v) \), and the maximum degree of \( G \) is given by \( \Delta = \max_{v \in V} \deg(v) \).

Dynamic graphs can be classified based on how their topological changes affect connectivity. The weakest meaningful requirement is temporal connectivity, which was introduced by Michail et al. \cite{MICHAIL_2016} in 2016. Consider a set $\mathcal{J}=\{(\overline{u_1u_2}, r_1),(\overline{u_2u_3}, r_2), \ldots, (\overline{u_ku_{k+1}}, r_k)\}$, where $\{\overline{u_1u_2}, ~\overline{u_2u_3},~ ..., ~\overline{u_ku_{k+1}}\}$ is a walk in $G$ (footprint) from $u_1$ to $u_{k+1}$ starting at round $r_1$ where $\forall \,i$, $\overline{u_iu_{i+1}}\in E_i$ and $r_{i+1}> r_i$. The set $\mathcal{J}$ is called a \emph{journey} from $u_1$ to $u_{k+1}$. All possible journeys from node $u$ to $v$ starting on or after round $r$ are denoted by $\mathcal{J}(u, v, r)$. Formally,

\begin{definition}
(temporal connectivity) \cite{MICHAIL_2016} A dynamic graph \( \mathcal{G} \) is said to be \emph{temporally connected} (or \emph{connected over time}) if for all \( r \in \mathbb{N}\cup \{0\} \) and for all pair of nodes \( u, v \in V \), there exists a journey from node \( u \) to node \( v \) starting at or after time \( r \). In other words, \( \mathcal{J}(u, v, r) \neq \emptyset \).\footnote{We refer to a dynamic graph that satisfies temporal connectivity as a \textit{temporally connected graph}.}
\end{definition}

Temporal connectivity is the minimal requirement for solving any global task in dynamic graphs, regardless of the initial positions of agents. In particular, any problem that requires every node to be involved, such as exploration or dispersion, is trivially unsolvable if the dynamic graph is not temporally connected. Stronger assumptions have also been studied in the literature. One well-known class of dynamic graphs guarantees connectivity at every round, rather than over time. A commonly used restriction is 1-interval connectivity, and a further refinement is its bounded variant.

\begin{definition}
 (\( \ell \)-bounded 1-interval connectivity) \cite{GOTOH_2021} A dynamic graph \( \mathcal{G} \) is \emph{1-interval connected} (or \emph{always connected}) if every snapshot \( \mathcal{G}_r \in \mathcal{S}_\mathcal{G} \) is connected. Furthermore, \( \mathcal{G} \) is said to be \( \ell \)-bounded 1-interval connected if it is always connected and \( |\overline{E}_r| \) is at most \( \ell \).
\end{definition}

\vspace{0.15cm}
\noindent\textbf{Agent}: We consider $ k$ agents to be present initially at nodes of the graph $G$. Each agent has a unique identifier assigned from the range $[1,\,n^{c}]$, where $c$ is a constant. Each agent knows its ID but is unaware of the other agents' IDs. Agents are not aware of the values of \( n \), \( k \), or $c$ unless stated otherwise. The agents are equipped with memory. An agent residing at a node, say $v$, in round $r$ knows $\deg(v)$ in the footprint $G$ and all associated ports corresponding to node $v$ in $G$; however, the agent does not understand if any incident edge of node $v$ is missing in $\mathcal{G}$ in round $t$. To be more precise, agent $a_i$ currently at node $v$ at round $r$ does not know the value of $\rho(e_v,r)$, where $e_v$ is an edge incident at node $v$. Such a model has been considered in \cite{GOTOH_2021}.

We call a node $v \in \mathcal{G}_r$ a \underline{hole} at round $r$ if it has no agent, and a \underline{multinode} if it has two or more agents.

\vspace{0.15cm}
\noindent \textbf{Communication model:} In this work, we consider an $l$-hop communication model, where $l \in \mathbb{N} \cup \{0\}$. Under this model, an agent at a node $v$ can send messages to all agents located within $l$ hops from $v$. When $l = 0$, this reduces to face-to-face (f-2-f) communication~\cite{Augustine_2018, Ajay_2019}, meaning agents can only communicate if they are co-located at the same node. At the other hand, when $l = \widetilde{D}$, $l$-hop communication becomes global communication~\cite{Ajay_dynamicdisp, Pelc_2006}, allowing agents to exchange messages regardless of their positions in the network.

\vspace{0.15cm}
\noindent \textbf{Visibility model:} In this work, we use two types of visibility models: 0-hop visibility and 1-hop visibility. In the 0-hop visibility model~\cite{Augustine_2018, Molla_2019, Molla_2020}, an agent at a node \( v \in \mathcal{G} \) can see the IDs of agents present at \( v \) in round \( r \), as well as the port numbers at \( v \), but nothing beyond that. In contrast, in the 1-hop visibility model~\cite{Agarwalla_2018, Avery_2020}, an agent \( a_i \) at node \( v \) can also see all neighbors of \( v \), including the IDs of agents (if any) at those neighboring nodes. Let \( e_v \) be an edge incident to node \( v \). In the 0-hop visibility model, agents cannot determine the value of \( \rho(e_v, r) \) at the beginning of round \( r \). In contrast, under the 1-hop visibility model, they can determine this value at the beginning of round \( r \).

% If the adversary sets $\rho(e_v,r)=0$ for an edge $e_v$ associated with a node $v$ with port $p$ at the start of a round $r$, the agent(s) present there cannot understand the value of $\rho(e_v,r)$ in the 0-hop visibility model. However, if the agent tries to move through $e_v$ at round $r$, it cannot make the move. In the 1-hop visibility model, agents can understand the value of $\rho(e_v,r)$ at the beginning of round $r$.

\vspace{0.2cm}
\noindent The algorithm runs in synchronous rounds. In each round $r$, an agent $a_i$ performs one \textit{Communicate-Compute-Move} (CCM) cycle as follows:
\begin{itemize}
    \item \textbf{Communicate:} Agent $a_i$ at node $v_i$ can communicate with other agents $a_j$ present at the same node $v_i$ or present at any different node $v_j$, depending on the communication (f-2-f or global). The agent also understands whether it had a \underline{successful} or an \underline{unsuccessful} move in the last round.
    \item \textbf{Compute:} Based on the information the agent has, the agent computes the port through which it will move or decides not to move at all.
    \item \textbf{Move:} Agent moves via the computed port or stays at its current node. 
\end{itemize}

In this work, we use two terms to describe agents' configurations within each round $r$: (i) \underline{at the beginning of round $r$}, which refers to the agents' configuration before the CCM cycle starts, and (ii) \underline{at the end of round $r$}, which refers to the agents' configuration after the CCM cycle completes.

If $k\geq 1$, and every node has at least \( \lfloor k/n \rfloor \) and at most \( \lceil k/n \rceil \) agents, then we refer to such a configuration as a \underline{dispersed} configuration in this work. Otherwise, we refer to it as an \underline{undispersed} configuration.

\section{Related work}\label{sec:rel}
The problem of dispersion for $k \leq n$ agents has been extensively studied in static graphs under two main communication models: (a) f-2-f communication, and (b) global communication. In f-2-f, there is a wide range of literature \cite{Augustine_2018, Ajay_2019, Molla_2019, ShintakuSKM20, KshemkalyaniS21, sudo, Molla_19, DAS_2023, Italiano_2025, Ajay_2025, Ajay_2025_, Patnayak_2025}. In 2022, Kshemkalyani et al. \cite{Ajay_2022} studied the dispersion problem for $k\leq n$ when agents are equipped with global communication. A natural and fundamental question is whether such communication capabilities are strictly necessary to achieve dispersion. In 2024, Gorain et al. \cite{Gorain_2024} partially addressed this question. They presented an algorithm for $k$ co-located agents that achieves dispersion without relying on f-2-f communication. In their model, the information available to an agent at a node $v$ at the beginning of each round is extremely limited. Specifically, each agent knows only the answers to the following two questions: (1) am I alone at node $v$? and (2) compared to the previous round, has the number of agents at $v$ increased or decreased? However, the question is still open: what is the necessary communication type to solve dispersion for $k\leq n$ agents?

Recently, researchers have shown growing interest in distributed problems such as exploration, gathering, dispersion, and leader election in dynamic graphs. This is motivated by the fact that dynamic graphs are more related to real-world networks, where topologies change over time. While dispersion has been well studied in static settings, relatively fewer works study the problem for $k \leq n$ agents in dynamic graphs~\cite{Saxena_2025_,Ajay_dynamicdisp,Agarwalla_2018}. Agarwalla et al.~\cite{Agarwalla_2018} consider a 1-bounded 1-interval connected dynamic ring and provide several deterministic algorithms. In contrast, Kshemkalyani et al.~\cite{Ajay_dynamicdisp} and Saxena et al.~\cite{Saxena_2025_} study dispersion in a more general and weaker dynamic model where the graph remains connected in every round, but there is no fixed footprint, and the adversary can arbitrarily add or remove edges.

In~\cite{Ajay_dynamicdisp}, two communication assumptions are considered: global communication and 1-neighborhood knowledge. This 1-neighbourhood knowledge is equivalent to saying 1-hop visibility. They proved the following two theorems:

\begin{theorem}\label{thm:1}
    \cite{Ajay_dynamicdisp} It is impossible to solve dispersion of $k \geq 5$ agents on a dynamic graph deterministically with the agents having 1-neighborhood knowledge and unlimited memory, but without global communication.
\end{theorem}

\begin{theorem}\label{thm:2}
\cite{Ajay_dynamicdisp} It is impossible to solve dispersion of $k \geq 3$ agents on a dynamic graph deterministically with the agents having global communication and unlimited memory, but without 1-neighborhood knowledge.
\end{theorem}

They provide an asymptotically optimal $\Theta(k)$ round algorithm in that setting. This communication model is strictly stronger than the one we consider. Translating their results into our model, one might infer that dispersion for $k\leq n$ agents in 1-interval connected dynamic graphs requires at least 1-hop visibility and 1-hop communication. This indicates a gap between their sufficient conditions and the necessary conditions in our model. Saxena et al.~\cite{Saxena_2025_} study on the same model \cite{Ajay_dynamicdisp} and extend the analysis to $T$-path connected graphs \footnote{A dynamic graph $\mathcal{G}$ is \emph{$T$-Path Connected} for $T \geq  1$, if for all $r \in \mathbb{N} \cup \{0\}$ and for any $u$, $v \in V$, there exists at least one round $i \in [r, r + T - 1]$ such that there exists a path between $u$ and $v$ in $\mathcal{G}_i$.}. Since $T$-path connectivity is weaker than 1-interval connectivity, Theorem \ref{thm:1} and Theorem \ref{thm:2} are valid in $T$-path connectivity as well. They design an algorithm using 1-hop visibility and global communication, and also show that dispersion becomes impossible in connectivity-time dynamic graphs even when agents have full visibility, global communication, and knowledge of all system parameters. Refer to Table~\ref{tab:imp} and Table~\ref{tab:algo} for the results from~\cite{Ajay_dynamicdisp, Saxena_2025_} in our model.

The dynamic graphs in both~\cite{Ajay_dynamicdisp, Saxena_2025_} are weaker than time-varying graphs (TVGs) because they lack an underlying footprint. Yet, their communication assumptions are quite restrictive. In contrast, we generalize the communication model by introducing $l$-hop communication, which captures a spectrum between f-2-f ($l = 0$) and global communication ($l = \widetilde{D}$). Their results identify conditions that are sufficient for dispersion with $k\leq n$ agents in weak dynamic models, but they do not establish what is necessary for stronger models like TVGs. They kept this as an open question in their conclusion. Our work fills this gap: we establish necessary conditions for solving $k$-balanced dispersion in temporally connected and $l$-bounded 1-interval connected TVGs. Since TVGs are stronger than the models used in~\cite{Ajay_dynamicdisp, Saxena_2025_}, our necessary conditions apply to those as well. While the question of necessary conditions for dispersion in static graphs remains open, we provide a necessary condition to solve $k$-balanced dispersion in TVGs.

\begin{table}[ht]
\centering
\scriptsize
\caption{Necessary conditions to solve $k$-balanced dispersion.}\label{imp}
\setlength{\tabcolsep}{6pt}
\renewcommand{\arraystretch}{2}
\begin{tabular}{|>{\centering\arraybackslash}p{2.5cm}|
                >{\centering\arraybackslash}p{3.5cm}|
                >{\centering\arraybackslash}p{3cm}|
                >{\centering\arraybackslash}p{2.8cm}|}
\hline
 \textbf{\raisebox{1pt}[2em][1em]{\makecell{Connectivity}}} & 
\textbf{\raisebox{1pt}[2em][1em]{\makecell{Number of agents}}} & 
\textbf{\raisebox{1pt}[2em][1em]{\makecell{Initial configuration}}} & 
\textbf{\raisebox{1pt}[2em][1em]{\makecell{Necessary \\conditions}}} \\
\hline \hline
\raisebox{1.2pt}[2.2em][1.5em]{\makecell{1-interval connected,\\ No footprint, \\ \cite{Ajay_dynamicdisp}}} & 
$5\leq k\leq n$ & 
\raisebox{1.2pt}[2.2em][1.5em]{\makecell{Any undispersed \\configuration}}  & 
\raisebox{1.2pt}[2.2em][1.5em]{\makecell{1-hop visibility,\\1-hop communication}} \\
\hline

\raisebox{1.2pt}[2.2em][1.5em]{\makecell{Connectivity time,\\ No footprint, \\ \cite{Saxena_2025_}}} & 
$3\leq k\leq n$ & 
\raisebox{1.2pt}[2.2em][1.5em]{\makecell{Any undispersed \\configuration}}  & 
\raisebox{1.2pt}[2.2em][1.5em]{\makecell{Impossible}} \\
\hline

\raisebox{1.2pt}[2.2em][1.5em]{\makecell{$T$-path connected,\\ No footprint, \\ \cite{Saxena_2025_}}} & 
$3\leq k\leq n$ & 
\raisebox{1.2pt}[2.2em][1.5em]{\makecell{Any undispersed \\configuration}}  & 
\raisebox{1.2pt}[2.2em][1.5em]{\makecell{1-hop visibility,\\1-hop communication}} \\
\hline

\raisebox{1pt}[2em][1.5em]{\makecell{Temporal\\connectivity \\ (TVGs)}} & \raisebox{1pt}[2em][1.5em]{\makecell{$k = pn + q$; $p,q \in \mathbb{N} \cup \{0\}$, \\ $3 \leq q \leq n - 1, n\geq 4$}}
 & \raisebox{1pt}[2em][1.5em]{\makecell{At least one node with \\ $\geq p+2$ agents}}
 & 
Impossible \\
\hline
\raisebox{1pt}[2em][1.5em]{\makecell{Temporal\\connectivity}} & 
\raisebox{1pt}[2em][1.5em]{\makecell{$k = pn + q$; $p,q \in \mathbb{N} \cup \{0\}$, \\ $3 \leq q \leq n - 3, n\geq 6$}} & \raisebox{1pt}[2em][1em]{\makecell{Any undispersed \\configuration}}
 & 
Impossible \\
\hline

\raisebox{1pt}[2em][1.5em]{\makecell{$1$-bounded\\1-interval connected }} & 
$k = pn$, $p \in \mathbb{N}$ & 
\raisebox{1pt}[2em][1.5em]{\makecell{Any undispersed \\configuration}}  & 
\raisebox{1pt}[2em][1.5em]{\makecell{1-hop visibility,\\1-hop communication}} \\
\hline

\raisebox{1pt}[2em][1.5em]{\makecell{$\ell$-bounded\\1-interval connected\\ ($\ell\geq 25$)}} & \raisebox{1pt}[2em][1.5em]{\makecell{$k = pn + q$; $p ,q\in \mathbb{N} \cup \{0\}$, \\ $5 \leq q \leq n - 1$, $n \in [5, \lfloor \sqrt{\ell} \rfloor]$}}
 & 
\raisebox{1pt}[2em][1.5em]{\makecell{At least one node with \\ $\geq p+2$ agents}} & 
\raisebox{1pt}[2em][1.5em]{\makecell{1-hop visibility,\\1-hop communication}}\\
\hline

\end{tabular}
\label{tab:imp}
\end{table}

\begin{table}[ht]
\centering
\scriptsize
\caption{Algorithmic results of $k$-balanced dispersion.}\label{algo}
\setlength{\tabcolsep}{6pt}
\renewcommand{\arraystretch}{3}
\begin{tabular}{|>{\centering\arraybackslash}p{1.9cm}|
                >{\centering\arraybackslash}p{1.8cm}|
                >{\centering\arraybackslash}p{1.9cm}|
                >{\centering\arraybackslash}p{2.8cm}|
                >{\centering\arraybackslash}p{2.9cm}|
                >{\centering\arraybackslash}p{2.6cm}|}
\hline
\textbf{\raisebox{1pt}[2em][1em]{\makecell{Connectivity}}}  & 
\textbf{\raisebox{1pt}[2em][1em]{\makecell{Number of \\agents}}} & 
\textbf{\raisebox{1pt}[2em][1em]{\makecell{Initial\\configuration}}}  &
\textbf{\raisebox{1pt}[2em][1em]{\makecell{Sufficient\\condition}}} & 
\textbf{\raisebox{1pt}[2em][1em]{\makecell{Time\\complexity}}} & 
\textbf{\raisebox{1pt}[2em][1em]{\makecell{Agent\\memory}}} \\ 
\hline\hline
\raisebox{1pt}[3em][2.5em]{\makecell{1-interval \\connected,\\ No footprint\\ \cite{Ajay_dynamicdisp}}} & 
$k \leq n$ & 
Any & 
\raisebox{1pt}[3em][2.5em]{\makecell{1-hop visibility,\\global communication}} & 
$\Theta(k)$ & 
$\Theta(\log n)$ \\
\hline
\raisebox{1pt}[3em][2.5em]{\makecell{$T$-path \\connected,\\ No footprint\\ \cite{Saxena_2025_}}} & 
$k \leq n$ & 
Any & 
\raisebox{1pt}[3em][2.5em]{\makecell{1-hop visibility,\\global communication,\\ knowledge of $T$}} & 
$\Theta(kT)$ & 
$\Theta(\log \max\{k,T\})$ \\
\hline
\raisebox{1pt}[3em][2.5em]{\makecell{$\ell$-bounded\\1-interval \\connected}} & 
$k \geq 1$ & 
Any & 
\raisebox{1pt}[3em][2.5em]{\makecell{1-hop visibility,\\global communication}} & 
$O(k)$ & 
$O(\log n)$ \\
\hline
\raisebox{1pt}[3em][2.5em]{\makecell{$1$-bounded\\1-interval \\connected}} & \raisebox{1pt}[3em][2.5em]{\makecell{$k = pn + q$,\\ $p \in \mathbb{N} \cup \{0\}$,\\ $q \in \{1,2\}$}}
 & 
\raisebox{1pt}[3em][2.5em]{\makecell{Co-located}} & 
\raisebox{1pt}[3em][2.5em]{\makecell{0-hop visibility,\\f-2-f communication,\\knowledge of $n$}}  & 
$O(n^4 \max(\log n, \log p))$ & 
$O(\max\{\log n, \log p\})$ \\
\hline

\end{tabular}
\label{tab:algo}
\end{table}

\section{Our contribution}
In this work, we introduce $k$-balanced dispersion, which is close to the load balancing problem. The details of our contributions are summarized as follows.

\begin{enumerate}
    \item Let $k = pn + q$; $p,q \in \mathbb{N} \cup \{0\}$, $3 \leq q \leq n - 1, n\geq 4$. We have shown that if initially there is at least one node in $G$ with $p+2$ agents, the problem of $k$-balanced dispersion is impossible to solve in temporally connected dynamic graphs. This holds even if nodes have infinite storage, and agents have infinite memory, full visibility, global communication, and complete knowledge of both $k$ and $n$. (Theorem \ref{thm:disp_temporal})
    \item Let $k = pn + q$; $p,q \in \mathbb{N} \cup \{0\}$, $3 \leq q \leq n - 3, n\geq 6$. We have shown that if the initial configuration is undispersed, the problem of $k$-balanced dispersion is impossible to solve in temporally connected dynamic graphs. This holds even if nodes have infinite storage, and agents have infinite memory, full visibility, global communication, and complete knowledge of both $k$ and $n$. (Theorem \ref{thm:disp_temporal_1})
    \item Let $k=pn$, $p\in \mathbb{N}$. If the initial configuration is undispersed, then to solve $k$-balanced dispersion in 1-bounded 1-interval connected graphs, agents need 1-hop visibility and 1-hop communication. (Theorem \ref{thm:imp_1hop} and Theorem \ref{thm:imp_global})

    \item Let $\ell\geq 25$ and $k = pn + q$; $p ,q\in \mathbb{N} \cup \{0\}$, $5 \leq q \leq n - 1$, $n \in [5, \lfloor \sqrt{\ell} \rfloor]$. If initially there is at least one node in $G$ with $p+2$ agents, then to solve $k$-balanced dispersion in $\ell$-bounded 1-interval connected graphs, agents need 1-hop visibility and 1-hop communication. (Theorem \ref{thm:imp_1hop_l>1} and Theorem \ref{thm:imp_global_l>1})

    \item We provide an algorithm which solves $k$-balanced dispersion in $\ell$-bounded 1-interval connected graphs with $k\geq 1$ agents in $O(k)$ rounds using $O(\log n)$ bits of memory per agent in the synchronous setting with global communication and 1-hop visibility (Section \ref{sec:l-bdd-1-int}).

    \item We provide an algorithm which solves $k$-balanced dispersion in $1$-bounded 1-interval connected graphs with $k\in \{pn+1, pn+2\}$ agents, $p\in \mathbb{N}$, in $O(n^4 \max(\log n, \log p))$ rounds using $O(\max\{\log n, \log p\})$ bits of memory per agent in the synchronous setting with global communication and 1-hop visibility (Section \ref{sec:1-bdd-1-int}).
\end{enumerate}

Refer to Table \ref{imp} for the necessary assumptions to solve $k$-balanced dispersion across different connectivity models, and refer to Table \ref{algo} for the sufficient assumptions to solve $k$-balanced dispersion in the same models.

\section{Impossibility for temporally connected graphs}
In this section, we prove that $k$-balanced dispersion is impossible to solve in temporally connected graphs. We begin by constructing a sequence of static graphs $\mathcal{S}_{\mathcal{G}} = \mathcal{G}_0,\, \mathcal{G}_1,\, \ldots,$ which is temporally connected graph. Then, we show that dispersion cannot be achieved in any round $r$; that is, it is impossible to solve dispersion in $\mathcal{G}_r$ for every $r \in \mathbb{N} \cup \{0\}$. We provide two results: (i) initial configuration has at least one node with more than $\lceil k/n\rceil$ agents, (ii) initial configuration is undispersed and has at most $\lceil k/n\rceil$ agents at every node.

\subsection{Impossibility with at least one node with more than $\lceil k/n\rceil$ agents}\label{sec:temp_1}
We consider a clique of size $n \geq 4$ as the footprint $G$ of the dynamic graph $\mathcal{G}$. The number of agents is $k = pn + q$, where $p,q \in \mathbb{N} \cup \{0\}$ and $3 \leq q \leq n - 1$. At round $r$, $\alpha_r(u)$ denotes the number of agents at $u$ at the beginning of round $r$, and $\beta_r(u)$ denotes the number of agents at the end of round $r$. At round $r$, we denote nodes as $v_1^r$, $v_2^r$, \ldots, $v_n^r$.

The key idea behind our construction is as follows. At every round $r$, the nodes are divided into two connected components: the first component, say $C_1^r$, includes at most $2$ nodes, and the second component, say $C_2^r$, includes the remaining nodes. The construction guarantees that if $C_1^r$ has a single node, it contains at least $p + 2$ agents; and if $C_1^r$ has two nodes, then it contains at least $2p + 3$ agents. In both cases, such construction ensures that there's a node in $C_1^r$ with at least $p + 2$ agents, which prevents successful dispersion. The construction of the dynamic graph $\mathcal{G}$ is as follows.
\begin{figure}
    \centering
    \includegraphics[width=0.27\linewidth]{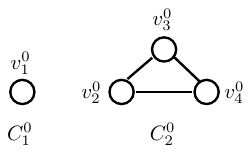}
    \caption{The graph $\mathcal{G}_0$.}
    \label{fig:round_0}
\end{figure}

\begin{figure}
    \centering
    \includegraphics[width=0.85\linewidth]{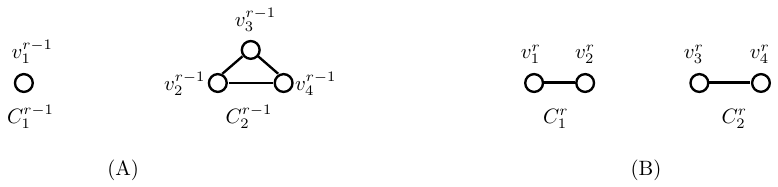}
     \caption{(A) The graph $\mathcal{G}_{r-1}$; (B) the graph $\mathcal{G}_{r}$, where $r$ is an odd round.}
    \label{fig:round_odd}
\end{figure}
\vspace{0.2cm}
\noindent \textbf{Construction of the dynamic graph $\mathcal{G}$:} Suppose $k$ agents are placed arbitrarily at nodes at round $r=0$. At round $0$, there is one node which has at least one node with at least $\lceil k/n\rceil+1= p+2$ agents. Without loss of generality, let $v_1^0\in \big\{v_i^0: \alpha_0(v_i^0)=\max\{\alpha_0(v_i^0): i\in [1,n]\}\big\}$. 

\begin{itemize}
    \item $\bm{r =0}$: Define the set of edges as $X_1=\big\{\overline{v_1^0 v_j^0} : j\in [2, n]\big\}$. At round $r$, we define $\rho(e,r)$ for every $e\in E$ as follows:
     \[
\rho(e,r) =
\begin{cases}
0, & \text{if } e\in X_1 \\
1, & \text{if } e\in E\setminus X_1
\end{cases}
\]

 The resulting graph has two connected components: the first, say $C_1^0$, is a clique of size 1, while the second, say $C_2^0$, is a connected component that is a clique of size $n-1$ (see Fig. \ref{fig:round_0} for $n=4$).

    \item $\bm{r=2i-1, i\in \mathbb{N}}$: At the end of round $r - 1$, the graph consists of two connected component $C_1^{r-1}$ and $C_2^{r-1}$, where $C_1^{r-1}$ is a clique of size 1, and $C_2^{r-1}$ is a clique of $n-1$ nodes (see Fig. \ref{fig:round_odd}(A) for $n=4$). Let node $v_1^{r-1}$ forms $C_1^{r-1}$, and nodes $v_2^{r-1}$, $v_3^{r-1}$, \ldots, $v_n^{r-1}$ forms $C_2^{r-1}$. Without loss of generality, let $v_2^{r-1}\in \big\{v_i^{r-1}: \beta_{r-1}(v_i^{r-1})=\max\{\beta_{r-1}(v_i^{r-1}): i\in [2,n]\}\big\}$. Define the set of edges as $X_2=\big\{\overline{v_i^{r-1} v_j^{r-1}} : i\in [1,2], \;\text{and }j\in [3, n]\big\}$. we define $\rho(e,r)$ for every $e\in E$ as follows:
     \[
\rho(e,r) =
\begin{cases}
0, & \text{if } e\in X_2 \\
1, & \text{if } e\in E\setminus X_2
\end{cases}
\]

The resulting graph has two connected components: the first, say $C_1^r$, is a clique of size 2, while the second, say $C_2^r$, is a connected component that is a clique of size $n-2$ (see Fig. \ref{fig:round_odd}(B) for $n=4$).
    
\begin{figure}
    \centering
    \includegraphics[width=0.85\linewidth]{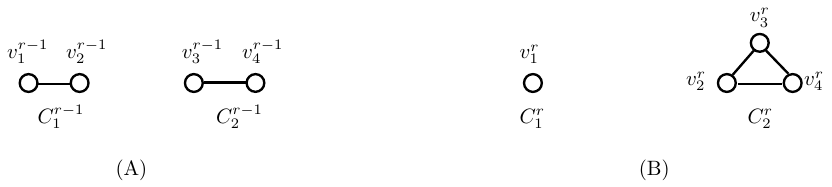}
     \caption{(A) The graph $\mathcal{G}_{r-1}$; (B) the graph $\mathcal{G}_{r}$, where $r$ is an even round and $r\geq 1$.}
    \label{fig:round_even}
\end{figure}
    \item $\bm{r=2i, i\in \mathbb{N}}$:  At the end of round $r - 1$, the graph consists of two connected component $C_1^{r-1}$ and $C_2^{r-1}$, where $C_1^{r-1}$ is a clique of two nodes, and $C_2$ is a clique of $n-2$ nodes (see Fig. \ref{fig:round_even}(A) for $n=4$). Let nodes $v_1^{r-1}$ and $v_2^{r-1}$ forms $C_1^{r-1}$, and nodes $v_3^{r-1}$, \ldots, $v_n^{r-1}$ forms $C_2^{r-1}$. Without loss of generality, let $v_1^{r-1}\in \big\{v_i^{r-1}: \beta_{r-1}(v_i^{r-1})=\max\{\beta_{r-1}(v_i^{r-1}): i\in \{1,2\}\}\big\}$. Define the set of edges as $X_3=\big\{\overline{v_1^{r-1} v_j^{r-1}} : j\in [2, n]\big\}$. we define $\rho(e,r)$ for every $e\in E$ as follows:
     \[
\rho(e,r) =
\begin{cases}
0, & \text{if } e\in X_3 \\
1, & \text{if } e\in E\setminus X_3
\end{cases}
\]

The resulting graph has two connected components: the first, say $C_1^r$, is a clique of size 1, while the second, say $C_2^r$, is a connected component that is a clique of size $n-1$ (see Fig. \ref{fig:round_even}(B) for $n=4$).
\end{itemize}

\begin{lemma}\label{lm:tempo_connec}
The dynamic graph $\mathcal{G}$ in our construction maintains temporal connectivity.
\end{lemma}

\begin{proof}
Consider any round $r \geq 0$ and any two nodes $u$ and $v$. If $u$ and $v$ are in the same connected component, then $\mathcal{J}(u, v, r) \neq 0$, since the component is a clique and $u$ and $v$ are directly connected. Suppose $u$ and $v$ are not in the same connected component. Without loss of generality, let node $u$ be in clique $C_1^r$, and node $v$ be in clique $C_2^r$. We divide the proof into two cases: Case (i) the size of clique $C_1^r$ is one, and Case (ii) the size of clique $C_1^r$ is two.

\begin{itemize}
    \item \textbf{Case (i):} As per the construction at round $r+1$, the adversary selects a node $w$ from $C_2^r$ and forms a clique $C_1^{r+1}$ with nodes $u$ and $w$. If $w = v$, then $\mathcal{J}(u, v, r) \neq 0$, since $C_1^{r+1}$ is a clique and $u$ and $v$ are directly connected. If $w \neq v$, then $v$ remains in $C_2^{r+1}$. At round $r+2$, the adversary selects a node $w'$ from $C_1^{r+1}$ and forms a clique $C_2^{r+2}$ with $w'$ and the nodes of $C_2^{r+1}$.

    \noindent \hspace{0.2cm}$\ast$ If $w' = u$, $\mathcal{J}(u, v, r) \neq 0$, as $u$ and $v$ are directly connected in the clique $C_2^{r+2}$.

    \noindent \hspace{0.2cm}$\ast$ If $w' \neq u$, $\mathcal{J}(u, v, r) \neq 0$, since there exists a journey from $u$ to $w'$ at round $r+1$ (as nodes $u$ and $w'$ are in clique $C_1^{r+1}$) and from $w'$ to $v$ at round $r+2$ (as nodes $w'$ and $v$ are in clique $C_2^{r+2}$), i.e., $\mathcal{J} = \{(\overline{u w'}, r+1), (\overline{w' v}, r+2)\}.$

    \item \textbf{Case (ii):} As per the construction at round $r+1$, the adversary selects a node $w$ from $C_1^{r}$ and forms a clique $C_2^{r+1}$ with $w$ and the nodes of $C_2^{r}$. If $w=u$, then $\mathcal{J}(u, v, r) \neq 0$, since $C_2^{r+1}$ is a clique and $u$ and $v$ are directly connected. If $w\neq u$, then it becomes Case (i) at round $r+1$. Since Case (i) is proved for every $r\geq 0$, therefore, $\mathcal{J}(u, v, r) \neq 0$.
\end{itemize}
In both cases, we have shown that $\mathcal{J}(u, v, r) \neq 0$. This completes the proof.
\end{proof}

\begin{theorem}\label{thm:disp_temporal}
In the constructed dynamic graph, for every round $r \geq 0$, there exists at least one node that has at least $p + 2$ agents. This holds even if nodes have infinite storage, and agents have infinite memory, full visibility, global communication, and complete knowledge of both $k$ and $n$.
\end{theorem}

\begin{proof}
Suppose there exists a deterministic algorithm \( \mathcal{A} \) that solves dispersion in our construction. In our construction, at round \( r = 0 \), node \( v_1^0 \) forms component \( C_1^0 \) and have at least \( p+2 \) agents. The remaining nodes \( v_2^0, \ldots, v_n^0 \) form a clique (component \( C_2^0 \)). Since no movement can occur between \( C_1^0 \) and \( C_2^0 \), the number of agents at node \( v_1 \) remains unchanged at the end of round 0. Therefore, the dispersion is not achieved at round $r=0$.

Without loss of generality, let \( v_2^0\) be the node with the maximum number of agents at the end of round 0 in $C_2^{0}$; that is, \( \beta_0(v_2^0) \geq \beta_0(v_i^0) \) for all \( i \in [3,n] \). As per our construction at round 1, we create two components: \( C_1^1 \) consisting of \( v_1^0 \) and \( v_2^0 \) (a clique of size 2), and \( C_2^1 \) consisting of the remaining nodes \( v_3^0, \ldots, v_n^0 \) (a clique of size \( n-2 \)). Let \( \beta_0(v_1^0) = p + q_1 \), \( q_1 \geq 2 \) since \( v_1^0 \) has at least \( p+2 \) agents. The total number of agents in \( C_2^0 \) is \( pn + q - (p + q_1) = p(n-1) + (q - q_1) \). Let $L=p(n-1) + (q - q_1)$. Then,

\[
\beta_0(v_2^0) \geq \frac{L}{n-1}
\]

Without loss of generality, let \( \beta_1(v_1^0) \geq \beta_1(v_2^0) \), we have the following inequality.

\begin{equation}
\label{eq:avg}
\beta_1(v_1^0) \geq \left\lceil \frac{\beta_0(v_1^0) + \beta_0(v_2^0)}{2} \right\rceil
\end{equation}

Using \( \beta_0(v_1^0) = p + q_1 \) and the bound on \( \beta_0(v_2^0) \), we derive the following from Eq.~\ref{eq:avg}.

\begin{equation}
\label{eq:ineq}
\beta_1(v_1^0) \geq p + \left\lceil \frac{(n-2)q_1 + q}{2(n-1)} \right\rceil
\end{equation}

Since \( q \geq 3 \) and \( q_1 \geq 2 \), we have the following inequality

\begin{align*}
\beta_1(v_1^0) &\geq p + \left\lceil \frac{2(n-2) + 3}{2(n-1)} \right\rceil =p+ \left\lceil 1+\frac{1}{2(n-1)} \right\rceil \implies \beta_1(v_1^0) \geq p + 2
\end{align*}

Hence, at the end of round \( r=1 \), there exists a node in \( C_1^1 \) with at least \( p+2 \) agents, so dispersion has not been achieved. At round \( r = 2 \), that node $v_1^0$ becomes isolated again, and forms a clique $C_1^2$ of size 1, like \( v_1^0 \) in round 0. Thus, the situation repeats. By applying this reasoning inductively, dispersion is never achieved in any round \( r \geq 0 \) as there is at least one node with $p+2$ agents.

Since this construction is independent of node storage, the agents' memory, visibility, communication model, or any knowledge, no such algorithm \( \mathcal{A} \) can exist. This completes the proof.
\end{proof}

\subsection{Impossibility with every node with at most $\lceil k/n\rceil $ agents}
We consider a clique of size $n \geq 6$ as the footprint $G$ of the dynamic graph $\mathcal{G}$. The number of agents is $k = pn + q$, where $p,q \in \mathbb{N} \cup \{0\}$ and $3 \leq q \leq n - 3$. At round $r$, $\alpha_r(u)$ denotes the number of agents at $u$ at the beginning of the round, and $\beta_r(u)$ denotes the number at the end. At round $r$, we denote nodes as $v_1^r$, $v_2^r$, \ldots, $v_n^r$.

The key idea behind our construction remains the same as in Section~\ref{sec:temp_1}, with a minor modification. We ensure that in every round, the dynamic graph consists of two connected components such that either some node has fewer than $p$ agents or some node has at least $p + 2$ agents. The construction of the dynamic graph $\mathcal{G}$ is as follows.
 
\vspace{0.2cm}
\noindent \textbf{Construction of the dynamic graph $\mathcal{G}$:} Suppose $k$ agents are placed arbitrarily at nodes at round $r=0$. At round $0$, there is at least one node which has at most $p-1$ agents. It is important that if $p=0$, then the undispersed configuration has a node with at least $p+2=2$ agents. Therefore, in this section, we can assume $p\geq 1$. For $p=0$, we can use the result from Section \ref{sec:temp_1}. So, for the rest of this section, we assume $p \geq 1$. Without loss of generality, let $v_1^0\in \big\{v_i^0: \alpha_0(v_i^0)=\min\{\alpha_0(v_i^0): i\in [1,n]\}\big\}$. 

\begin{itemize}
    \item $\bm{r =0}$: Define the set of edges as $X_1=\big\{\overline{v_1^0 v_j^0} : j\in [2, n]\big\}$. At round $r$, we define $\rho(e,r)$ for every $e\in E$ as follows:
     \[
\rho(e,r) =
\begin{cases}
0, & \text{if } e\in X_1 \\
1, & \text{if } e\in E\setminus X_1
\end{cases}
\]

 The resulting graph has two connected components: the first, say $C_1^0$, is a clique of size 1, while the second, say $C_2^0$, is a connected component that is a clique of size $n-1$.

    \item $\bm{r=2i-1, i\in \mathbb{N}}$: At the end of round $r - 1$, the graph consists of two connected component $C_1^{r-1}$ and $C_2^{r-1}$, where $C_1^{r-1}$ is a clique of size 1, and $C_2^{r-1}$ is a clique of $n-1$ nodes. Let node $v_1^{r-1}$ forms $C_1^{r-1}$, and nodes $v_2^{r-1}$, $v_3^{r-1}$, \ldots, $v_n^{r-1}$ forms $C_2^{r-1}$. If there is a node with at least $p+2$ agents in $C_2^{r-1}$ at the end of round $r-1$, then the adversary considers round $r=2i-1$ as $r=0$, and starts using the construction mentioned in Section \ref{sec:temp_1}. Since we have already shown in Section \ref{sec:temp_1} that at every round $r$, there is a node with at least $p+2$ agents, therefore dispersion is not achieved for every round $r\geq 2i-1$. Otherwise, if there is no node with at least $p+2$ agents in $C_2^{r-1}$ at the end of round $r-1$, then it constructs $\mathcal{G}_r$ as follows.
    
    Without loss of generality, let $v_2^{r-1}\in \big\{v_i^{r-1}: \beta_{r-1}(v_i^{r-1})=\min\{\beta_{r-1}(v_i^{r-1}): i\in [2,n]\}\big\}$. Define the set of edges as $X_2=\big\{\overline{v_i^{r-1} v_j^{r-1}} : i\in [1,2], \;\text{and }j\in [3, n]\big\}$. we define $\rho(e,r)$ for every $e\in E$ as follows:
     \[
\rho(e,r) =
\begin{cases}
0, & \text{if } e\in X_2 \\
1, & \text{if } e\in E\setminus X_2
\end{cases}
\]

The resulting graph has two connected components: the first, say $C_1^r$, contains a clique of size 2, while the second, say $C_2^r$, is a connected component that is a clique of size $n-2$.

    \item $\bm{r=2i, i\in \mathbb{N}}$:  At the end of round $r - 1$, the graph consists of two connected component $C_1^{r-1}$ and $C_2^{r-1}$, where $C_1^{r-1}$ is a clique of two nodes, and $C_2$ is a clique of $n-2$ nodes. Let nodes $v_1^{r-1}$ and $v_2^{r-1}$ forms $C_1^{r-1}$, and nodes $v_3^{r-1}$, \ldots, $v_n^{r-1}$ forms $C_2^{r-1}$. Without loss of generality, let $v_1^{r-1}\in \big\{v_i^{r-1}: \beta_{r-1}(v_i^{r-1})=\min\{\beta_{r-1}(v_i^{r-1}): i\in \{1,2\}\}\big\}$. Define the set of edges as $X_3=\big\{\overline{v_1^{r-1} v_j^{r-1}} : j\in [2, n]\big\}$. we define $\rho(e,r)$ for every $e\in E$ as follows:
     \[
\rho(e,r) =
\begin{cases}
0, & \text{if } e\in X_3 \\
1, & \text{if } e\in E\setminus X_3
\end{cases}
\]

The resulting graph has two connected components: the first, say $C_1^r$, is a clique of size 1, while the second, say $C_2^r$, is a connected component that is a clique of size $n-1$.
\end{itemize}

\begin{lemma}
The dynamic graph $\mathcal{G}$ in our construction maintains temporal connectivity.
\end{lemma}

\begin{proof}
The proof is similar to Lemma \ref{lm:tempo_connec}.
\end{proof}

\begin{theorem}\label{thm:disp_temporal_1}
In the constructed dynamic graph, $k$-balanced dispersion is not achieved at any round $r \geq 0$. This holds even if nodes have infinite storage, and agents have infinite memory, full visibility, global communication, and complete knowledge of both $k$ and $n$.
\end{theorem}

\begin{proof}
Suppose there exists a deterministic algorithm \( \mathcal{A} \) that solves dispersion in our construction. In our construction, at round \( r = 0 \), node \( v_1^0 \) forms component \( C_1^0 \) and have at most \( p-1 \) agents. The remaining nodes \( v_2^0, \ldots, v_n^0 \) form a clique (component \( C_2^0 \)). Since no movement can occur between \( C_1^0 \) and \( C_2^0 \), the number of agents at node \( v_1 \) remains unchanged at the end of round 0. Therefore, the dispersion is not achieved at round $r=0$.

If at the end of round $r=0$, there exists a node in $C_2^0$ with at least $p+2$ agents, the adversary treats round $r = 1$ as the new starting round $r = 0$ and proceeds with the construction described in Section~\ref{sec:temp_1}. As shown in Theorem~\ref{thm:disp_temporal}, this ensures that in every round $r$, at least one node has at least $p+2$ agents, preventing dispersion for all $r \geq 1$. It remains to handle the case where no node in $C_2^0$ has at least $p+2$ agents at the end of round $0$, i.e., every node has at most $p+1$ agents. We now show that even in this case, dispersion is not achieved at the end of round $r = 1$.

Without loss of generality, let \( v_2^0\) be the node with the least number of agents at the end of round 0 in $C_2^{0}$; that is, \( \beta_0(v_2^0) \leq \beta_0(v_i^0) \) for all \( i \in [3,n] \). As per our construction at round 1, we create two components: \( C_1^1 \) consisting of \( v_1^0 \) and \( v_2^0 \) (a clique of size 2), and \( C_2^1 \) consisting of the remaining nodes \( v_3^0, \ldots, v_n^0 \) (a clique of size \( n-2 \)). Let \( \beta_0(v_1^0) = p - q_1 \), \( q_1 \geq 1 \) since \( v_1^0 \) has fewer than \( p \) agents. The value of $\beta_0(v_2^0)$ has two choices, either $\leq p$ or $p+1$, as every node of $C_2^0$ has at most $p+1$ agents at the end of round $r=0$. We discuss each case for the value of $\beta_0(v_2^0)$ one by one as follows.

\begin{itemize}
    \item If $\beta_0(v_2^0)\leq p$, then at the end of round $r=1$ in $C_1^1$, there is a node which has $<p$ agents as $\beta_0(v_2^0)\leq p$ and $\beta_0(v_1^0)<p$.
    \item If $\beta_0(v_2^0)=p+1$, then at the end of round $r=0$, $\beta_0(v_i^0)=p+1$ for every $i\in [3,n]$. Therefore, the total agents in $C_2^0$ is $(p+1)(n-1)$. The value of $k$ is nothing but the sum of agents in $C_1^0$ and agents in $C_2^0$. 

   \begin{align*}
       np+q&=(n-1)(p+1)+p-q_1\implies q=n-1-q_1\\
       \implies n-3&\geq n-1-q_1 \;\; \text{as }\;\; q\leq n-3  \\
       \therefore\;\; q_1\geq 2 
   \end{align*}

   Since at the beginning of round $r=1$ in $C_1^1$, $\beta_0(v_1^0)\leq p-2$, and  $\beta_0(v_2^0)=p+1$, at the end of round $r=1$ in $C_1^r$, there is a node which has fewer than $p$ agents.
\end{itemize}

Hence, at the end of round \( r=1 \), there exists a node, say $v$, in \( C_1^1 \) with fewer than \( p \) agents, so dispersion has not been achieved. At round \( r = 2 \), node $v$ becomes isolate. Thus, the situation repeats. By applying this reasoning inductively, dispersion is never achieved in any round \( r \geq 0 \).

Since this construction is independent of node storage, the agents' memory, visibility, communication model, or any knowledge, no such algorithm \( \mathcal{A} \) can exist. This completes the proof.
\end{proof}

Based on Theorem \ref{thm:disp_temporal} and Theorem \ref{thm:disp_temporal_1}, we have the following theorem.

 \begin{theorem}\label{thm:final_temporal}
    The problem of $k$-balanced dispersion is impossible to solve in temporally connected dynamic graphs. This holds even if nodes have infinite storage, and agents have infinite memory, full visibility, global communication, and complete knowledge of both $k$ and $n$.
\end{theorem}

\section{Impossibility and lower bound for $\ell$-bounded 1-interval connected graphs}\label{sec:imp_l-bdd}
In this section, we provide impossibility results of the dispersion problem in $\ell$-bounded 1-interval connected graphs, where $\ell\geq 1$.

\begin{figure}
    \centering
    \includegraphics[width=0.27\linewidth]{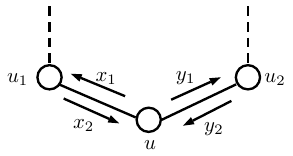}
    \caption{Movement of agents at round $r$ at nodes $u$, $u_1$, and $u_2$.}
    \label{fig:ring}
\end{figure}

\begin{theorem}\label{thm:imp_1hop}  
    It is impossible to solve $k$-balanced dispersion in 1-bounded 1-interval connected graphs for any $n \geq 4$ if agents have only 1-hop visibility and f-2-f communication. The impossibility holds for all undispersed initial configurations of agents, even when nodes have infinite storage, agents have infinite memory and know the values of $n$ and $k$.
\end{theorem}

\begin{proof}
    Consider the ring of size $n$ as a footprint $G$. Let $k=pn$ agents be in an undispersed configuration in $G$, where $p\in \mathbb{N}$.\footnote{If $k = pn$, then in a dispersed configuration, every node has $p$ agents.} If the agents have 1-hop visibility, then as per our visibility model, the agent can sense the missing edge corresponding to some port $\lambda$. The movement of the agents depends on 1-hop visibility and f-2-f communication. We prove our result using a contrapositive argument. Suppose there is a deterministic algorithm $\mathcal{A}$ which solves dispersion.

Since the adversary knows the algorithm $\mathcal{A}$, it can precompute the agents' movement in $\mathcal{G}_r$ for every round $r$. At round $r$, considering $\rho(e, r) = 1$ for all $e \in E$, if it precomputes that it does not lead to a dispersed configuration, then it keeps $\mathcal{G}_r=G$ at the beginning of round $r$. Otherwise, it sets $\rho(e_j, r) = 0$ for a chosen edge $e_j \in E$, and $\rho(e, r) = 1$ for every $e\in E\setminus \{e_j\}$. How the adversary selects such an edge $e_j$ at round $r$ is described below.

Since agents are not in a dispersed configuration at the beginning of round $r$, there is a node $u$ which has at most $p-1$ agents. Let $u_1$ and $u_2$ be two neighbours of node $u$ in $G$. Let $z$ ($\leq p-1$) agents be at node $u$. If $\mathcal{G}_r=G$, then the following is the precomputation of the agents as per algorithm $\mathcal{A}$: let $x_1$ agents move from node $u$ to $u_1$, $y_1$ agents move from node $u$ to $u_2$, $x_2$ agents move from node $u_1$ to $u$, and $y_2$ agents move from node $u_2$ to $u$ (see Fig. \ref{fig:ring}). And accordingly, at the end of round $r$, there are $z+(x_2-x_1)+(y_2-y_1)=p$ agents at node $u$ in $G$. The adversary knows the precomputed values of $x_1, x_2, y_1, y_2$. It is important to note that $(x_2-x_1)\leq 0$ and $(y_2-y_1)\leq 0$ are not possible due to the following reason. If $(x_2-x_1)\leq 0$ and $(y_2-y_1)\leq 0$, then the value $z+(x_2-x_1)+(y_2-y_1)\leq z<p$ as $z<p$. Therefore, either $(x_2-x_1)> 0$ or $(y_2-y_1)> 0$. Without loss of generality, let $(y_2-y_1)> 0$. In this case, at the beginning of round $r$, the adversary constructs $\mathcal{G}_r$ as follows: $\rho(\overline{uu_2},r)=0$, and $\rho(e,r)=1$ for every $e\in E\setminus \{\overline{uu_2}\}$. Let's call this new configuration $G'$. Note that the 1-hop view of agents at node $u_1$ in $G'$ is the same as the 1-hop view at node $u_1$ in $G$, and the 1-hop view of agents at node $u$ and $u_2$ changes at round $r$. Since f-2-f communication is present, agents at node $u_1$ can not get the information about the configuration change, and hence the movement of agents at node $u_1$ in $G'$ remains the same as in $G$ at round $r$. Since the 1-hop view of agents at node $u$ and $u_2$ has changed, they can change their decision based on algorithm $\mathcal{A}$. In this case, if more than $x_1$ or fewer than $x_1$ agents move from $u$ to $u_1$ in $G'$, then at node $u_1$, the number of agents is not $p$ due to the fact that the 1-hop view of node $u_1$ in $G'$ is the same as in $G$, and $n\geq 4$. Therefore, the number of agents at node $u$ in $G'$ is $z+(x_2-x_1)$. If $z+(x_2-x_1)=p$, then $z+(x_2-x_1)=z+(x_2-x_1)+(y_2-y_1) \implies (y_2-y_1)=0$. This gives a contradiction. Therefore, at the end of round $r$, the number of agents at node $u$ in $G'$ is not $p$. Hence, the dispersion is not achieved in $G'$ at the end of round $r$.   

Since in every round $r\geq 0$, $\overline{E}_r\leq 1$, and in every round $\mathcal{G}_r$ is connected. Therefore, $\mathcal{G}$ is a 1-bounded 1-interval connected graph. Also, this proof is independent of nodes having infinite storage, agents having infinite memory and knowledge of $n$, $k$. This completes the proof. 
\end{proof}

\begin{theorem}\label{thm:imp_global}
    It is impossible to solve $k$-balanced dispersion in 1-bounded 1-interval connected graphs for any $n \geq 3$ if agents have only 0-hop visibility and global communication. The impossibility holds for all undispersed initial configurations of agents, even when agents have infinite memory and know the values of $n$ and $k$.
\end{theorem}

\begin{proof} 
Consider the ring of size $n$ as a footprint $G$. Let $k=pn$ agents be in an arbitrary configuration at nodes of $G$, where $p\in \mathbb{N}$. If the agents have 0-hop visibility, then the choice of movement of each agent depends only on the information present at the current node and the information attained via global communication. Therefore, the movement of the agents depends on 0-hop visibility and global communication. We prove our result using a contrapositive argument. Suppose there is a deterministic algorithm $\mathcal{A}$ which solves dispersion. 
    
Since the adversary knows the algorithm $\mathcal{A}$, it can precompute the agents' movement in $\mathcal{G}_r$ for every round $r$. At round $r$, considering $\rho(e, r) = 1$ for all $e \in E$, if it precomputes that it does not lead to a dispersed configuration, then it keeps $\mathcal{G}_r=G$ at the beginning of round $r$. Otherwise, it sets $\rho(e_j, r) = 0$ for a chosen edge $e_j \in E$, and $\rho(e, r) = 1$ for every $e\in E\setminus \{e_j\}$. How the adversary selects such an edge $e_j$ at round $r$ is described below.

Since agents are not in a dispersed configuration at the beginning of round $r$, there is a node $u$ which has at least $p+1$ agents. Let $u_1$ and $u_2$ be two neighbours of node $u$ in $G$. Let $z$ ($\geq p+1$) agents be at node $u$. If $\mathcal{G}_r=G$, then the following is the precomputation of the agents as per algorithm $\mathcal{A}$: let $x_1$ agents move from node $u$ to $u_1$, $y_1$ agents move from node $u$ to $u_2$, $x_2$ agents move from node $u_1$ to $u$, and $y_2$ agents move from node $u_2$ to $u$ (see Fig. \ref{fig:ring}). And accordingly, at the end of round $r$, there are $z+(x_2-x_1)+(y_2-y_1)=p$ agents at node $u$. The adversary knows the precomputed values of $x_1, x_2, y_1, y_2$. It is important to note that $(x_2-x_1)\geq 0$ and $(y_2-y_1)\geq 0$ are not possible due to the following reason. If $(x_2-x_1)\geq 0$ and $(y_2-y_1)\geq 0$, the value of $z+(x_2-x_1)+(y_2-y_1)\geq z\geq p+1$. Therefore, either $(x_2-x_1)< 0$ or $(y_2-y_1)< 0$. Without loss of generality, let $(y_2-y_1)< 0$. In this case, at the beginning of round $r$, the adversary constructs $\mathcal{G}_r$ as follows: $\rho(\overline{uu_2},r)=0$, and $\rho(e,r)=1$ for every $e\in E\setminus \{\overline{uu_2}\}$. Let's call this new configuration $G'$. Since agents have 0-hop visibility, they can not understand the configuration change at round $r$. Hence, the movement of agents in $G'$ is the same as in $G$. Since edge $\overline{uu_2}$ is not present, the movement of agents from node $u$ to $u_2$ and node $u_2$ to $u$ in $G'$ will be unsuccessful. Therefore, the number of agents at node $u$ in $G'$ is $z+(x_2-x_1)$. If $z+(x_2-x_1)=p$, then $z+(x_2-x_1)=z+(x_2-x_1)+(y_2-y_1) \implies (y_2-y_1)=0$. This gives a contradiction as $y_2<y_1$. Therefore, at the end of round $r$, the number of agents at node $u$ in $G'$ is not $p$. Hence, the dispersion is not achieved in $G'$ at the end of round $r$.   

Since in every round $r\geq 0$, $\overline{E}_r\leq 1$, and in every round $\mathcal{G}_r$ is connected. Therefore, $\mathcal{G}$ is a 1-bounded 1-interval connected graph. Also, this proof is independent of nodes having infinite storage, agents having infinite memory and knowledge of $n$, $k$. This completes the proof. 
\end{proof}

\begin{remark}
    Theorems~\ref{thm:imp_1hop} and~\ref{thm:imp_global} hold when $k = pn$, for any $p \in \mathbb{N}$. That is, a team of $k = pn$ agents requires both 1-hop visibility and 1-hop communication to solve $k$-balanced dispersion in 1-bounded 1-interval connected graphs.
\end{remark}

\begin{remark}
    Theorems~\ref{thm:imp_1hop} and~\ref{thm:imp_global} are valid in $\ell$-bounded 1-interval connected graphs for any $\ell \geq 1$. In $\ell$-bounded 1-interval connected graphs, solving $k$-balanced dispersion with $k = pn$ agents (where $p \in \mathbb{N}$) also requires both 1-hop visibility and 1-hop communication.
\end{remark}

Now, we show two impossibilities to solve $k$-balanced dispersion in $\ell$-bounded 1-interval connected graphs. 

\begin{theorem}\label{thm:imp_1hop_l>1}
\textbf{($\ell \geq 25$)} It is impossible to solve $k$-balanced dispersion in $\ell$-bounded 1-interval connected graphs when agents have $1$-hop visibility and f-2-f communication. The impossibility holds for all initial configurations where at least one node has more than $\lceil k/n \rceil$ agents, even if nodes have infinite storage, agents have infinite memory and knowledge of $n$, $\ell$, and $k$.
\end{theorem}

\begin{proof}
      In this proof, we consider $\ell\geq 25$ and the clique of size $n$, where $n \in \big[5, \lfloor \sqrt{l} \rfloor \big]$, as a footprint $G$. Consider \( k = pn + q \), where \( p,q \in \mathbb{N} \cup \{0\} \) and \( 5 \leq q \leq n - 1 \). Consider an initial configuration which has at least one node containing at least $\lceil k/n \rceil+1=p+2$ agents. Let $S_{\mathcal{G}_r}(u)$ denote the number of agents at node $u$ at the beginning of round $r$, and $E_{\mathcal{G}_r}(u)$ denote the number of agents at node $u$ at the end of round $r$. At the beginning of round $r\geq 0$, the adversary finds a path $w_1 \sim w_2\sim \ldots\sim w_n$ of length $n$ in $G$ such that $S_{\mathcal{G}_r}(w_i)\geq S_{\mathcal{G}_r}(w_{i+1})$ holds for every $i\in [1,n-1]$. Since $G$ is a clique, the adversary can always find at least one such path in $G$ at the beginning of round $r$. If there are several such paths, the adversary can pick any. Let's call the path chosen by the adversary as \( P_r \). Let $E'$ be the edges of $P_r$. At the beginning of round \( r \), the adversary sets \( \mathcal{G}_r \) to either \( P_r \) or \( P_r' \), where \( P_r' \) is defined later.

      Suppose there is an algorithm $\mathcal{A}$ that solves dispersion. Since the adversary is aware of $\mathcal{G}_r$ and the agent's algorithm, it can precompute the movement of agents at round $r$. Based on precomputation, if the adversary observes that at the end or round $r$, there is at least one node in $P_r$ which has at least $p+2$ agents, then it constructs $\mathcal{G}_r=P_r$ at the beginning of round $r$. Otherwise, based on precomputation, if the adversary observes that at the end or round $r$, there is no node in $P_r$ which has at least $p+2$ agents, it sets $\mathcal{G}_r$ to $P_r'$ such that there is at least one node in $P_r'$ with at least $p+2$ agents at the end of round $r$. Let $E''=\big(E'\setminus \{\overline{w_4w_5}\}\big)\cup \{\overline{w_1w_5}\}$. The graph $P_r''$ is $(V, E'')$.

 \begin{figure}[ht!]
    \centering
    \includegraphics[width=0.75\linewidth]{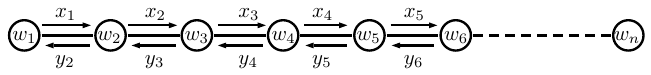}
    \caption{Agents' movement when \( \mathcal{G}_r \) is \( P_r \).}
    \label{fig:P_r}
\end{figure}
 
 Observe that in $P_r'$, the 1-hop view of agents located at nodes $w_1$, $w_4$, and $w_5$ is changed, while the 1-hop view of agents at all other nodes remains identical to their view in $P_r$, due to the presence of f-2-f communication. Now we prove that there is a node in $P_r'$ which has at least $p+2$ agents. Since there is a node with at least $p+2$ agents in $\mathcal{G}_r$ at the beginning of round $r$, and $S_{\mathcal{G}_r}(w_i) \geq S_{\mathcal{G}_r}(w_{i+1})$ for every $i$, therefore, $S_{P_r}(w_1) \geq p + 2$. As per precomputaion in $P_r$, let $x_i$ agents move from node $w_i$ to $w_{i+1}$, and $y_{i+1}$ agents move from node $w_{i+1}$ to $w_i$ (see Fig. \ref{fig:P_r} for agents' movement in $P_r$). Since at the end of round $r$, there is no node with at least $p+2$ agents in $P_r$, the following cases are possible. 

\begin{itemize}
    \item \textbf{Case 1:} $E_{P_r}(w_3)=p-c_1$ and $E_{P_r}(w_4)=p-c_2$, where $c_1\geq 0$ and $c_2\geq 0$,
    \item \textbf{Case 2:} $E_{P_r}(w_3)=p+1$ and $E_{P_r}(w_4)=p-c_2$, where $c_2\geq 0$,
    \item \textbf{Case 3:} $E_{P_r}(w_3)=p-c_1$ and $E_{P_r}(w_4)=p+1$, where $c_1\geq 0$,
    \item \textbf{Case 4:} $E_{P_r}(w_3)=p+1$ and $E_{P_r}(w_4)=p+1$.
\end{itemize}

Since $q\geq 5$, $S_{P_r}(w_1)\geq p+2$ and $S_{P_r}(w_i)\geq S_{P_r}(w_{i+1})$ in $P_r$, the value of $\sum_{i=1}^{4}S_{P_r}(w_i)$ is at least $4p+5$. Hence, 

$$x_4-y_5=\sum_{i=1}^{4}S_{P_r}(w_i)-\sum_{i=1}^{4}E_{P_r}(w_i)\geq4p+5-\sum_{i=1}^{4}E_{P_r}(w_i). $$

Since at the end of round $r$, there is no node in $P_r$ which has at least $p+2$ agents, $E_{P_r}(w_1)\leq p+1$ and $E_{P_r}(w_2)\leq p+1$. Hence,

 $$x_4-y_5\geq 2p+3-\sum_{i=3}^{4}E_{P_r}(w_i).$$ 

Since in all aforementioned cases, we have the value of $E_{P_r}(w_3)$ and $E_{P_r}(w_4)$, the value of $x_4-y_5$ is:

\begin{itemize}
    \item  Case 1 $\implies$ $x_4-y_5\geq 3+c_1+c_2$
    \item Case 2 $\implies$ $x_4-y_5\geq 2+c_2$,
    \item Case 3 $\implies$ $x_4-y_5\geq 2+c_1$,
    \item Case 4 $\implies$ $x_4-y_5\geq 1$.
\end{itemize}

\begin{figure}
    \centering
    \includegraphics[width=0.75\linewidth]{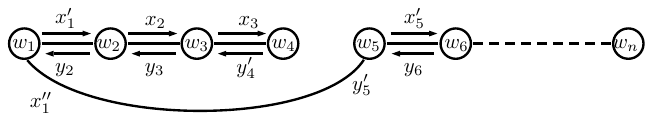}
    \caption{Agents' movement when \( \mathcal{G}_r \) is \( G' \).}
    \label{fig:1-hop}
\end{figure}

Note that if $\mathcal{G}_r=P_r$, then $E_{P_r}(w_3)=S_{P_r}(w_3)+(y_4+x_2)-(y_3+x_3)$, and $E_{P_r}(w_4)=S_{P_r}(w_4)+(y_5+x_3)-(y_4+x_4)$. Since the 1-hop view of node $w_2$ and $w_3$ in $P_r'$ is the same as in $P_r$, the decision of agents' movement at node $w_2$ and $w_3$ in $P_r'$ remains the same as in $P_r$. Note that in $P_r'$, $S_{P_r}(w_i)=S_{P_r'}(w_i)$ for every $i$. Since there is no edge between node $w_4$ and $w_5$, no agent can move from node $w_4$ to $w_5$ in $P_r'$. Therefore, in $P_r'$, $E_{P_r'}(w_4)=S_{P_r'}(w_4)+x_3-y_4'$ and $E_{P_r'}(w_3)=S_{P_r'}(w_3)+(y_4'+x_2)-(y_3+x_3)$, where $y_4'$ is the number agents that move from node $w_4$ to node $w_3$ in $P_r'$ as the 1-hop view of $w_2$, $w_3$ are the same but the 1-hop view of node $w_4$ is changed (See Fig. \ref{fig:1-hop} for agents' movement in $P_r'$). Observe that if \( x_4 - x_5 \leq 0 \), then the value of \( \sum_{i=1}^{4}E_{P_r}(w_i) \) is at least \( \sum_{i=1}^{4}S_{P_r}(w_i) \). Therefore, \( \sum_{i=1}^{4}E_{P_r}(w_i) \geq 4p + 5 \) as \( \sum_{i=1}^{4}S_{P_r}(w_i) \geq 4p + 5 \). This contradicts the fact that at the end of round \( r \), at most \( p + 1 \) agents are present at every node in \( P_r \). So, the value of \( x_4 - x_5 \) must be at least 1. In \( P_r' \), one might argue that if \( x_1'' - y_5' = x_4 - y_5 \) holds, where \( x_1'' \) agents move from node \( w_1 \) to \( w_5 \) in $P_r'$, and \( y_5' \) agents move from node \( w_5 \) to \( w_1 \) in $P_r'$, then the extra load of agents can be distributed across nodes \( w_5, w_6, \ldots, w_n \), and it might be possible to ensure that at the end of round \( r \), at most \( p + 1 \) agents are present at every node in \( P_r' \). Note that agents at nodes \( w_1 \), \( w_4 \), and \( w_5 \) are aware of this change. However, agents at nodes \( w_2 \) and \( w_3 \) are not aware of it; as previously mentioned, the decisions of agents at \( w_2 \) and \( w_3 \) remain the same as in \( P_r \). Below, we show that due to this reason, there is a node with at least \( p + 2 \) agents at the end of round \( r \) in \( P_r' \).

\begin{itemize}
    \item \textbf{Case 1:} Here, $x_4-y_5\geq 3+c_1+c_2$. We divide this case into the following sub-cases.
    \begin{itemize}
        \item If $y_4'\leq y_4+1+c_1$,
        \begin{align*}
            E_{P_r'}(w_4)&=S_{P_r'}(w_4)+x_3-y_4'\\
            &=S_{P_r}(w_4)+x_3-y_4',  \; \text{ as } \;S_{P_r}(w_i)=S_{P_r'}(w_i)\\
            E_{P_r'}(w_4) &\geq S_{P_r}(w_4)+x_3-y_4-1-c_1,  \; \text{ as } \; y_4'\leq y_4+c_1+1\\
             &= S_{P_r}(w_4)+x_3-y_4 + (y_5-x_4)-(y_5-x_4)-1-c_1\\
             &=S_{P_r}(w_4)+(y_5+x_3)-(y_4+x_4)+(x_4-y_5)-1-c_1\\
              &=  p+(x_4-y_5)-1-c_1-c_2, \;\; \\
              &\;\;\text{ as } \;S_{P_r}(w_4)+(y_5+x_3)-(y_4+x_4)=E_{P_r}(w_4)=p-c_2\\
              \therefore E_{P_r'}(w_4)&\geq p+2, \; \text{ as } \;x_4 - y_5 \geq 3+c_1+c_2\\
        \end{align*}   
        \item If $y_4'\geq y_4+2+c_1$,
        \begin{align*}
            E_{P_r'}(w_3)&=S_{P_r'}(w_3)+(y_4'+x_2)-(y_3+x_3)\\
            &=S_{P_r}(w_3)+(y_4'+x_2)-(y_3+x_3), \; \text{ as } \;S_{P_r}(w_i)=S_{P_r'}(w_i)\\
            E_{P_r'}(w_3) &\geq S_{P_r}(w_3)+(y_4+x_2)-(y_3+x_3)+c_1+2, \; \text{ as } \;y_4'\geq y_4+c_1+2\\
              &= p+2, \;\text{ as }\; S_{P_r}(w_3)+(y_4+x_2)-(y_3+x_3)=E_{P_r}(w_3)=p-c_1 \\
              \therefore E_{P_r'}(w_3)&\geq p+2
        \end{align*}
    \end{itemize}
    \item \textbf{Case 2:} In this case, $x_4-y_5\geq 2+c_2$. We divide this case into the following sub-cases. 
    \begin{itemize}
        \item If $y_4'\leq y_4$,
        \begin{align*}
           E_{P_r'}(w_4)&=S_{P_r'}(w_4)+x_3-y_4'\\
            &=S_{P_r}(w_4)+x_3-y_4',  \; \text{ as } \;S_{P_r}(w_i)=S_{P_r'}(w_i)\\
            E_{P_r'}(w_4) &\geq S_{P_r}(w_4)+x_3-y_4,  \; \text{ as } \; y_4'\leq y_4\\
             &= S_{P_r}(w_4)+x_3-y_4 + (y_5-x_4)-(y_5-x_4)\\
             &=S_{P_r}(w_4)+(y_5+x_3)-(y_4+x_4)+(x_4-y_5)\\
              &=  p+(x_4-y_5)-c_2 \;\; \\
              &\text{ as } \;S_{P_r}(w_4)+(y_5+x_3)-(y_4+x_4)=E_{P_r}(w_4)=p-c_2\\
              \therefore E_{P_r'}(w_4)&\geq p+2, \; \text{ as } \;x_4 - y_5 \geq 2+c_2\\
        \end{align*}
        \item If $y_4'\geq y_4+1$,
                \begin{align*}
            E_{P_r'}(w_3)&=S_{P_r'}(w_3)+(y_4'+x_2)-(y_3+x_3)\\
            &=S_{P_r}(w_3)+(y_4'+x_2)-(y_3+x_3), \; \text{ as } \;S_{P_r}(w_i)=S_{P_r'}(w_i)\\
            E_{P_r'}(w_3) &\geq S_{P_r}(w_3)+(y_4+x_2)-(y_3+x_3)+1, \; \text{ as } \;y_4'\geq y_4+1\\
              &= p+2, \;\text{ as }\; S_{P_r}(w_3)+(y_4+x_2)-(y_3+x_3)=E_{P_r}(w_3)=p+1 \\
              \therefore E_{P_r'}(w_3)&\geq p+2
        \end{align*}
    \end{itemize}
    \item \textbf{Case 3:} In this case, $x_4-y_5\geq 2+c_1$. We divide this case into the following sub-cases. 

     \begin{itemize}
        \item If $y_4'\leq y_4+c_1+1$,
        \begin{align*}
           E_{P_r'}(w_4)&=S_{P_r'}(w_4)+x_3-y_4'\\
            &=S_{P_r}(w_4)+x_3-y_4'  \; \text{ as } \;S_{P_r}(w_i)=S_{P_r'}(w_i)\\
            E_{P_r'}(w_4) &\geq S_{P_r}(w_4)+x_3-y_4-c_1-1  \; \text{ as } \; y_4'\leq y_4+c_1+1\\
             &= S_{P_r}(w_4)+x_3-y_4 + (y_5-x_4)-(y_5-x_4)-c_1-1\\
             &=S_{P_r}(w_4)+(y_5+x_3)-(y_4+x_4)+(x_4-y_5)-c_1-1\\
              &=  p+(x_4-y_5)-c_1 \;\; \text{ as } \;S_{P_r}(w_4)+(y_5+x_3)-(y_4+x_4)=S_{P_r}(w_4)=p+1\\
              \therefore E_{P_r'}(w_4)&\geq p+2 \; \text{ as } \;x_4 - y_5 \geq 2+c_1\\
        \end{align*}
        \item If $y_4'\geq y_4+c_1+2$,
                \begin{align*}
             E_{P_r'}(w_3)&=S_{P_r'}(w_3)+(y_4'+x_2)-(y_3+x_3)\\
            &=S_{P_r}(w_3)+(y_4'+x_2)-(y_3+x_3), \; \text{ as } \;S_{P_r}(w_i)=S_{P_r'}(w_i)\\
            E_{P_r'}(w_3) &\geq S_{P_r}(w_3)+(y_4+x_2)-(y_3+x_3)+c_1+2, \; \text{ as } \;y_4'\geq y_4+c_1+2\\
              &= p+2, \;\text{ as }\; S_{P_r}(w_3)+(y_4+x_2)-(y_3+x_3)=E_{P_r}(w_3)=p-c_1 \\
              \therefore E_{P_r'}(w_3)&\geq p+2
        \end{align*}
    \end{itemize}
    \item \textbf{Case 4:} In this case, $x_4-y_5\geq 1$. We divide this case into the following sub-cases.
    \begin{itemize}
        \item If $y_4'\leq y_4$,
        \begin{align*}
            E_{P_r'}(w_4)&= S_{P_r'}(w_4)+x_3-y_4'\\
            &= S_{P_r}(w_4)+x_3-y_4', \; \text{ as } \;S_{P_r}(w_i)=S_{P_r'}(w_i)\\
            E_{P_r'}(w_4) &\geq S_{P_r}(w_4)+x_3-y_4,  \; \text{ as } \; y_4'\leq y_4\\
             &= S_{P_r}(w_4)+x_3-y_4 + (y_5-x_4)-(y_5-x_4)\\
             &=S_{P_r}(w_4)+(y_5+x_3)-(y_4+x_4)+(x_4-y_5)\\
              E_{P_r'}(w_4)&\geq  p+1+(x_4-y_5), \;\; \text{ as } \;S_{P_r}(w_4)+(y_5+x_3)-(y_4+x_4)=E_{P_r}(w_4)=p+1\\
              \therefore\; E_{P_r'}(w_4)&\geq p+2, \; \text{ as } \;x_4 - y_5 \geq 1
        \end{align*}
        \item If $y_4'\geq y_4+1$,
                \begin{align*}
            E_{P_r'}(w_3)&=S_{P_r'}(w_3)+(y_4'+x_2)-(y_3+x_3)\\
            &=S_{P_r}(w_3)+(y_4'+x_2)-(y_3+x_3), \; \text{ as } \;S_{P_r}(w_i)=S_{P_r'}(w_i)\\
            E_{P_r'}(w_3)&\geq S_{P_r}(w_3)+(y_4+x_2)-(y_3+x_3)+1, \; \text{ as } \; y_4'\geq y_4+1\\
             \therefore\; E_{P_r'}(w_3)&\geq p+2, \;\;\;\text{ as }\; S_{P_r}(w_3)+(y_4+x_2)-(y_3+x_3)=E_{P_r}(w_3)=p+1 
        \end{align*}
    \end{itemize}
\end{itemize}

In all cases, either node $w_3$ or $w_4$ has at least $p+2$ agents in $P_r'$ at the end of round $r$. Therefore, the dispersion is not achieved. Since in every round $r\geq 0$, $\mathcal{G}_r$ is connected. Therefore, $\mathcal{G}$ is a 1-interval connected graph. At every round $r$, since $|E_r|=n-1$ and $|E|=\frac{n(n-1)}{2}$, $\overline{E}_r=|E|-|E_r|=\frac{n(n-1)}{2}-n+1=\frac{n(n-2)+2}{2}\leq n^2\leq\ell$. Therefore, $\mathcal{G}$ is $\ell$-bounded 1-interval connected graph. Also, this proof is independent of nodes having infinite storage, agents having infinite memory and knowledge of $\ell$, $n$, $k$. This completes the proof. 
\end{proof}

\begin{theorem}\label{thm:imp_global_l>1}
    (\textbf{Impossibility without 1-hop visibility and $\ell\geq 25$}) It is impossible to solve $k$-balanced dispersion in $\ell$-bounded 1-interval connected graphs when agents have global communication but no 1-hop visibility. The impossibility holds for all initial configurations where at least one node has more than $\lceil k/n \rceil$ agents, even if nodes have infinite storage, agents have infinite memory and knowledge of $n$, $\ell$, and $k$.
\end{theorem}
\begin{proof}
  In this proof, we consider $\ell\geq 25$ and the clique of size $n$, where $n \in \big[5, \lfloor \sqrt{l} \rfloor \big]$, as a footprint $G$. Consider \( k = pn + q \), where \( p,q \in \mathbb{N} \cup \{0\} \) and \( 5 \leq q \leq n - 1 \). Consider an initial configuration which has at least one node containing at least $\lceil k/n \rceil+1=p+2$ agents. Let $S_{\mathcal{G}_r}(u)$ denote the number of agents at node $u$ at the beginning of round $r$, and $E_{\mathcal{G}_r}(u)$ denote the number of agents at node $u$ at the end of round $r$. At the beginning of round $r\geq 0$, the adversary finds a path $w_1 \sim w_2\sim \ldots\sim w_n$ of length $n$ in $G$ such that $S_{\mathcal{G}_r}(w_i)\geq S_{\mathcal{G}_r}(w_{i+1})$ holds for every $i\in [1,n-1]$. Since $G$ is a clique, the adversary can always find at least one such path in $G$ at the beginning of round $r$. If there are several such paths, the adversary can pick any. Let's call the path chosen by the adversary as \( P_r \). Let $E'$ be the edges of $P_r$. At the beginning of round \( r \), the adversary sets \( \mathcal{G}_r \) to either \( P_r \) or \( P_r' \), where \( P_r' \) is defined later.

     Suppose there is an algorithm $\mathcal{A}$ that solves dispersion. Since the adversary is aware of $\mathcal{G}_r$ and the agent's algorithm, it can precompute the movement of agents at round $r$. Based on precomputation, if the adversary observes that at the end or round $r$, there is at least one node in $P_r$ which has at least $p+2$ agents, then it constructs $\mathcal{G}_r=P_r$ at the beginning of round $r$. Otherwise, based on precomputation, if the adversary observes that at the end or round $r$, there is no node in $P_r$ which has at least $p+2$ agents, it sets $\mathcal{G}_r$ to $P_r'$ such that there is at least one node in $P_r'$ with at least $p+2$ agents at the end of round $r$. Let $E''=\big(E'\setminus \{\overline{w_4w_5}\}\big)\cup \{\overline{w_1w_5}\}$. The graph $P_r''$ is $(V, E'')$.

 Observe that in $P_r'$, the 0-hop view of agents at all nodes remains identical to their view in $P_r$, due to the absence of 1-hop visibility. Now we prove that given that there is no node with $p+2$ agents at the end of round $r$ in $P_r$, there is a node in $P_r'$ which has at least $p+2$ agents. Since there is a node with at least $p+2$ agents in $\mathcal{G}_r$ at the beginning of round $r$, and $S_{\mathcal{G}_r}(w_i) \geq S_{\mathcal{G}_r}(w_{i+1})$ for every $i$, therefore, $S_{P_r'}(w_1) \geq p + 2$. As per precomputaion in $P_r$, let $x_i$ agents moves from node $w_i$ to $w_{i+1}$, and $y_{i+1}$ agents move from node $w_{i+1}$ to $w_i$ (see Fig. \ref{fig:P_r} for agents' movement in $P_r$). The following cases are possible at the end of round $r$.

 \begin{figure}[ht!]
    \centering
    \includegraphics[width=0.75\linewidth]{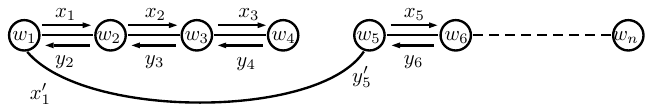}
    \caption{Agents' movement when \( \mathcal{G}_r \) is \( G' \).}
    \label{fig:global}
\end{figure}

\begin{itemize}
    \item \textbf{Case 1:} $E_{P_r}(w_3)=p-c_1$ and $E_{P_r}(w_4)=p-c_2$, where $c_1\geq 0$ and $c_2\geq 0$,
    \item \textbf{Case 2:} $E_{P_r}(w_3)=p+1$ and $E_{P_r}(w_4)=p-c_2$, where $c_2\geq 0$,
    \item \textbf{Case 3:} $E_{P_r}(w_3)=p-c_1$ and $E_{P_r}(w_4)=p+1$, where $c_1\geq 0$,
    \item \textbf{Case 4:} $E_{P_r}(w_3)=p+1$ and $E_{P_r}(w_4)=p+1$.
\end{itemize}

Since $q\geq 5$, $S_{P_r}(w_1)\geq p+2$ and $S_{P_r}(w_i)\geq S_{P_r}(w_{i+1})$ in $P_r$, the value of $\sum_{i=1}^{4}S_{P_r}(w_i)$ is at least $4p+5$. Hence,  

$$x_4-y_5=\sum_{i=1}^{4}S_{P_r}(w_i)-\sum_{i=1}^{4}E_{P_r}(w_i)\geq4p+5-\sum_{i=1}^{4}E_{P_r}(w_i). $$

Since at the end of round $r$, there is no node in $P_r$ which has at least $p+2$ agents, $E_{P_r}(w_1)\leq p+1$ and $E_{P_r}(w_2)\leq p+1$. Hence,

 $$x_4-y_5\geq 2p+3-\sum_{i=3}^{4}E_{P_r}(w_i).$$ 

Since in all aforementioned cases, we have the value of $E_{P_r}(w_3)$ and $E_{P_r}(w_4)$, the value of $x_4-x_4$ is:

\begin{itemize}
    \item  Case 1 $\implies$ $x_4-y_5\geq 3+c_1+c_2$
    \item Case 2 $\implies$ $x_4-y_5\geq 2+c_2$,
    \item Case 3 $\implies$ $x_4-y_5\geq 2+c_1$,
    \item Case 4 $\implies$ $x_4-y_5\geq 1$.
\end{itemize}

Note that if $\mathcal{G}_r=P_r$, then $E_{P_r}(w_3)=S_{P_r}(w_3)+(y_4+x_2)-(y_3+x_3)$, and $E_{P_r}(w_4)=S_{P_r}(w_4)+(y_5+x_3)-(y_4+x_4)$. Since the 0-hop view of agents in $P_r'$ is the same as the 0-hop view, the decision of agents' movement at every node in $P_r'$ remains the same as in $P_r$ (See Fig. \ref{fig:global} for agents' movement in $P_r'$). Since there is no edge between node $w_4$ and $w_5$, no agent can move from node $w_4$ to $w_5$ in $P_r'$. Therefore, in $P_r'$, $E_{P_r'}(w_4)=S_{P_r'}(w_4)+x_3-y_4$ and $E_{P_r'}(w_3)=S_{P_r'}(w_3)+(y_4+x_2)-(y_3+x_3)$. Observe that if \( x_4 - x_5 \leq 0 \), then the value of \( \sum_{i=1}^{4}E_{P_r}(w_i) \) is at least \( \sum_{i=1}^{4}S_{P_r}(w_i) \). Therefore, \( \sum_{i=1}^{4}E_{P_r}(w_i) \geq 4p + 5 \) as \( \sum_{i=1}^{4}S_{P_r}(w_i) \geq 4p + 5 \). This contradicts the fact that at the end of round \( r \), at most \( p + 1 \) agents are present at every node in \( P_r \). So, the value of \( x_4 - x_5 \) must be at least 1. Note that at node \( w_i \), if an agent attempts to move through some port \( \lambda \) in \( P_r \), but there is no outgoing edge via port \( \lambda \), the move fails. For instance, some agents at nodes \( w_1 \) and \( w_5 \) may attempt to move through ports that lead to \( w_5 \) and \( w_1 \), respectively. Since the edge between \( w_1 \) and \( w_5 \) does not exist in \( P_r \), these moves fail. However, in \( P_r' \), the edge \( \overline{w_1w_5} \) appears, so the same movements become successful. For this reason, in Fig.~\ref{fig:global}, we use \( x_1' \) to denote a move from \( w_1 \) to \( w_5 \), and \( y_5' \) to denote a move from \( w_5 \) to \( w_1 \). In \( P_r' \), one might argue that if \( x_1' - y_5' = x_4 - y_5 \) holds, where \( x_1' \) agents move from node \( w_1 \) to \( w_5 \), and \( y_5 \) agents move from node \( w_5 \) to \( w_1 \), then the extra load of agents can be distributed across nodes \( w_5, w_6, \ldots, w_n \), and it might be possible to ensure that at the end of round \( r \), at most \( p + 1 \) agents are present at every node in \( P_r' \). Note that each agent in $P_r'$ has the same 0-hop view as in $P_r$; as previously mentioned, the decisions of agents at \( w_2 \), \( w_3 \), $w_4$ remain the same as in \( P_r \). Below, we show that due to this reason, there is a node with at least \( p + 2 \) agents at the end of round \( r \) in \( P_r' \).

\begin{itemize}
    \item \textbf{In Case 1:} In this case, $x_4-y_5\geq 3+c_1+c_2$. 
        \begin{align*}
           E_{P_r'}(w_4)&=S_{P_r'}(w_4)+x_3-y_4\\
            &=S_{P_r}(w_4)+x_3-y_4,  \; \text{ as } \;S_{P_r}(w_i)=S_{P_r'}(w_i)\\
             &= S_{P_r}(w_4)+x_3-y_4 + (y_5-x_4)-(y_5-x_4)\\
             &=S_{P_r}(w_4)+(y_5+x_3)-(y_4+x_4)+(x_4-y_5)\\
              &=  p+(x_4-y_5)-c_2, \;\; \text{ as } \;S_{P_r}(w_4)+(y_5+x_3)-(y_4+x_4)=E_{P_r}(w_4)=p-c_2\\
               E_{P_r'}(w_4)&\geq p+2+c_1, \; \text{ as } \;x_4 - y_5 \geq 3+c_1+c_2\\
              \therefore E_{P_r'}(w_4)&\geq p+2, \; \text{ as } \; c_1\geq 0
        \end{align*}
    \item \textbf{In Case 2:} In this case, $x_4-y_5\geq 2+c_2$.      
        \begin{align*}
            E_{P_r'}(w_4)&=S_{P_r'}(w_4)+x_3-y_4\\
            &=S_{P_r}(w_4)+x_3-y_4,  \; \text{ as } \;S_{P_r}(w_i)=S_{P_r'}(w_i)\\
             &= S_{P_r}(w_4)+x_3-y_4 + (y_5-x_4)-(y_5-x_4)\\
             &=S_{P_r}(w_4)+(y_5+x_3)-(y_4+x_4)+(x_4-y_5)\\
              &=  p+(x_4-y_5)-c_2, \;\;\text{ as } \;S_{P_r}(w_4)+(y_5+x_3)-(y_4+x_4)=E_{P_r}(w_4)=p-c_2\\
              \therefore E_{P_r'}(w_4)&\geq p+2, \; \text{ as } \;x_4 - y_5 \geq 2+c_2
        \end{align*}
    \item \textbf{In Case 3:} In this case, $x_4-y_5\geq 2$. 
        \begin{align*}
            E_{P_r'}(w_4)&=S_{P_r'}(w_4)+x_3-y_4\\
            &=S_{P_r}(w_4)+x_3-y_4,  \; \text{ as } \;S_{P_r}(w_i)=S_{P_r'}(w_i)\\
             &= S_{P_r}(w_4)+x_3-y_4 + (y_5-x_4)-(y_5-x_4)\\
             &=S_{P_r}(w_4)+(y_5+x_3)-(y_4+x_4)+(x_4-y_5)\\
              &=  p+1+(x_4-y_5), \;\; \text{ as } \;S_{P_r}(w_4)+(y_5+x_3)-(y_4+x_4)=S_{P_r}(w_4)=p+1\\
              E_{P_r'}(w_4)&\geq p+3+c_1, \; \text{ as } \;x_4 - y_5 \geq 2+c_1\\
              \therefore E_{P_r'}(w_4)&\geq p+3, \; \text{ as } \; c_1\geq 0
        \end{align*}
\item \textbf{In Case 4:} In this case, $x_4-y_5\geq 1$.
\begin{align*}
            E_{P_r'}(w_4)&= S_{P_r'}(w_4)+x_3-y_4\\
            &= S_{P_r}(w_4)+x_3-y_4, \; \text{ as } \;S_{P_r}(w_i)=S_{P_r'}(w_i)\\
             &= S_{P_r}(w_4)+x_3-y_4 + (y_5-x_4)-(y_5-x_4)\\
             &=S_{P_r}(w_4)+(y_5+x_3)-(y_4+x_4)+(x_4-y_5)\\
              &=  p+1+(x_4-y_5), \;\; \text{ as } \;S_{P_r}(w_4)+(y_5+x_3)-(y_4+x_4)=E_{P_r}(w_4)=p+1\\
              \therefore\; E_{P_r'}(w_4)&\geq p+2, \; \text{ as } \;x_4 - y_5 \geq 1
        \end{align*}
\end{itemize}
In all cases, either node $w_3$ or $w_4$ has at least $p+2$ agents in $P_r'$ at the end of round $r$. Therefore, the dispersion is not achieved. Since in every round $r\geq 0$, $\mathcal{G}_r$ is connected. At every round $r$, since $|E_r|=n-1$ and $|E|=\frac{n(n-1)}{2}$, $\overline{E}_r=|E|-|E_r|=\frac{n(n-1)}{2}-n+1=\frac{n(n-2)+2}{2}\leq n^2\leq\ell$. Therefore, $\mathcal{G}$ is $\ell$-bounded 1-interval connected graph. Also, this proof is independent of nodes having infinite storage, agents having infinite memory and knowledge of $\ell$, $n$, $k$. This completes the proof.  
\end{proof}

\begin{remark}
    Theorems~\ref{thm:imp_1hop_l>1} and~\ref{thm:imp_global_l>1} hold for $\ell \geq 25$. A team of $k = pn + q$ agents, where $p,q \in \mathbb{N} \cup \{0\}$, $5 \leq q \leq n - 1$, and $n \in [5, \lfloor \sqrt{\ell} \rfloor]$, requires both 1-hop visibility and 1-hop communication to solve $k$-balanced dispersion in $\ell$-bounded 1-interval connected graphs.

    In both Theorems~\ref{thm:imp_1hop_l>1} and~\ref{thm:imp_global_l>1}, we assume that there is at least one node with more than $\lceil k/n \rceil$ agents. As noted in Section~\ref{sec:intro}, it is possible that even when a node has exactly $\lceil k/n \rceil$ agents, $k$-balanced dispersion has not yet been achieved. We leave this case as an open question. 
\end{remark}

Now, we provide the memory lower bound result to solve $k$-balanced dispersion in static graphs for $k \geq n+1$. The proof of this lemma is motivated by \cite{Augustine_2018}. In \cite{Augustine_2018}, they have provided the memory lower bound $\Omega(\log k)$ to solve dispersion on a static graph when $k\leq n$. We extend their idea for the memory lower bound to solve dispersion in static graphs $G$ when $k\geq n+1$.
\begin{lemma}\label{lm:memory_lower}
        Assuming all agents are given the same amount of memory, agents require $\Omega(\log n)$ bits of memory each for any deterministic algorithm to solve $k$-balanced dispersion, where $k \geq  n+1$.
\end{lemma}
\begin{proof}
    Suppose all agents have $o(\log n)$ bits of memory. Each agent’s state space is $2^{o(\log n)}=n^{o(1)}$. Since there are at least $n+1$ agents, by the pigeonhole principle, there exist at least $k-n+2$ agents with the same state space. Let agents $a_1$, $a_2$, \ldots, $a_{k-n+2}$ have the same state space. Let us suppose that all agents are initially co-located at some node. Since all agents run the same deterministic algorithm, and agents $a_1$, $a_2$, \ldots, $a_{k-n+2}$ are co-located initially, they will run the same moves. It means that agents $a_1$, $a_2$, \ldots, $a_{k-n+2}$ remain stay together. Therefore, by the end of the algorithm, agents $a_1$, $a_2$, \ldots, $a_{k-n+2}$ are at the same node. Solving dispersion requires a configuration where there is at least one agent at every node. We will never solve dispersion because there is no way for agents $a_1$, $a_2$, \ldots, $a_{k-n+2}$ to settle down on at least two different nodes. Therefore, all $k$ agents are at most $n-1$ nodes by the end of the algorithm. Thus, to solve the dispersion problem, agents need $\Omega(\log n)$ memory each. This completes the proof.
\end{proof}

\section{$k$-balanced dispersion in $\ell$-bounded 1-interval connected graphs}\label{sec:l-bdd-1-int}
In this section, we present an algorithm to solve $k$-balanced dispersion in $\ell$-bounded 1-interval connected graphs, assuming agents are equipped with $1$-hop visibility and global communication. The algorithm in \cite{Ajay_dynamicdisp} solves dispersion for $k \leq n$ agents, but it works in a model without a footprint, where port numbers are assigned dynamically based on the current degree of a node. In contrast, our model assumes a fixed footprint graph $G$, so port numbers remain fixed throughout. In this section, we will extend the algorithm of \cite{Ajay_dynamicdisp} to solve $k$-balanced dispersion in $\ell$-bounded 1-interval connected graphs when agents are equipped with 1-hop visibility and global communication. Let’s refer to the algorithm in \cite{Ajay_dynamicdisp} as \( Algo\_weak\_disp(k) \). First, we recall the idea of their algorithm.

\subsection{High-level overview of $Algo\_weak\_disp(k)$}
 The algorithm proceeds in synchronous rounds, and each round has two key steps. First, the agents identify the connected components of $\mathcal{G}_r$ that contain at least one agent. Second, they perform a \emph{slide} operation, where agents try to fill a hole. Each agent is equipped with 1-hop visibility and global communication. At the beginning of each round, the minimum ID agent $a_j$ at node $v$ broadcasts all the information which it gets using 1-hop visibility. As a result, each agent can reconstruct a consistent partial snapshot of the current graph $\mathcal{G}_r$. This reconstructed view consists of subgraphs $H_1, H_2, \ldots, H_x$, where there is at least one agent at nodes of $H_i$ for all $i$, and there are no edges in $\mathcal{G}_r$ that connect nodes of $H_i$ and $H_j$ for all $i\neq j$. 

Now, suppose in some component $H_i$ there exists a multinode $v_m$, and within $H_i$ there is a path $P = v_1, v_2, \ldots, v_l$, where $v_1 = v_m$, each node from $v_2$ to $v_l$ have exactly one agent, and one of the neighbors of $v_l$, say $v_h$, is a hole. In this configuration, a coordinated slide can take place: the agent at $v_1$ moves to $v_2$, the agent at $v_2$ moves to $v_3$, and so on, until the agent at $v_l$ moves to the hole $v_h$. Due to this slide, the number of agents of multinode(s) decreases, and it eventually achieves dispersion for $k\leq n$ agents.

Based on this, they have the following theorem.

\begin{theorem}\label{thm:main_result}
   \cite{Ajay_dynamicdisp} Given $k \leq n$ agents placed arbitrarily on the nodes of any $n$-node graph $\mathcal{G}_r$, where $\mathcal{G}_r$ is the graph at round $r\geq 0$.
  \( Algo\_weak\_Disp(k) \) solves the dispersion in $\Theta(k)$ rounds with $\Theta(\log k)$ bits at each agent in the synchronous setting with global communication and 1-hop visibility. Also, all agents understand that dispersion has been achieved and terminated. 
\end{theorem}

The authors have the following lemma, which is used to prove Theorem \ref{thm:main_result}.
\begin{lemma}\label{lm:decrement_holes}
  \cite{Ajay_dynamicdisp} Consider any $n$-node dynamic graph $\mathcal{G}_r$ at round $r \geq 0$. If $k \leq n$ agents are positioned on $l$ nodes of $\mathcal{G}_r$ at the beginning of round $r$, then at the beginning of round $r+1$, the agents are positioned on at least $l + 1$ nodes of $\mathcal{G}_r$. 
\end{lemma}

\subsection{Our algorithm with Correctness}
Let $a_i$ be the agent with the smallest ID agent at node $v$ in $\mathcal{G}_r$ during round $r$. For each port label $q \in \{0, 1, \ldots, \deg(v) - 1\}$, let $u$ be the neighbor of $v$ accessible via port $q$, and let $ID_v = a_i.ID$. Agent $a_i$ constructs and broadcasts the information $S_v$ using global communication, where $S_v$ is defined as follows. If node $u$ contains one or more agents, define the tuple $S_v^q = \left(\alpha(v),\, ID_v,\, q,\, ID_u\right),$ where $\alpha(v)$ denotes the number of agents at node $v$, and $ID_u$ is the smallest ID of at node $u$. If node $u$ is a hole, define the tuple as $S_v^q = \left(\alpha(v),\, ID_v,\, q,\, \perp\right),$ where $\perp$ indicates that port $q$ leads to an unoccupied node. Based on these, $S_v = \left\{S_v^q \;\middle|\; 0 \leq q < \deg(v)\right\}.$ In $Algo\_weak\_disp(k)$, minimum ID agent at node $v$ broadcasts $S_v$ at every round $r$. The detailed description of our algorithm is as follows.

Each agent $a_i$ maintains two parameters.

\begin{itemize}
    \item $a_i.ID$: In this parameter, it stores its ID.
    \item $a_i.P$: This parameter is a dummy variable that takes the value either $0$ or $1$. This parameter agents use to understand whether each node has at least one agent or not. Initially, at round $0$, $a_i.P = 0$.
\end{itemize}

 At round $r$, agent $a_i$ at node $v$ does the following. 

\begin{itemize}
    \item \textbf{If $\bm{a_i.P=0}$:} 
    
    \noindent If agent $a_i$ is the minimum ID agent at node $v$, then it does the following.

\begin{itemize}
    \item If it finds $S_u^q = \left(\alpha(u),\, ID_u,\, q,\, \perp\right)$ for some node $u$, or $S_v^q = \left(\alpha(v),\, ID_v,\, q,\, \perp\right)$, then it does the following.

    \begin{itemize}
        \item If it finds some tuple $S_{u_1}$ with $\alpha(u_1) > 1$ for some node $u_1$, or if $\alpha(v) > 1$, then it executes $Algo\_weak\_disp(k)$.
        \item Otherwise, it terminates.
    \end{itemize}
    
    \item If it does not find any such $S_u^q = \left(\alpha(u),\, ID_u,\, q,\, \perp\right)$ for some node $u$ and $S_v^q = \left(\alpha(v),\, ID_v,\, q,\, \perp\right)$, then it sets $a_i.P=1$.
\end{itemize}

\noindent If agent $a_i$ is not the minimum ID agent at node $v$, then it does the computation for the decision of the minimum ID agent at $v$, say $a_l$. It follows $a_l$'s decision as follows:

\begin{itemize}
    \item If $a_l$ decides to execute $Algo\_weak\_disp(k)$, then $a_i$ does not do anything in this round.
    \item If $a_l$ terminates, then $a_i$ also terminates.
    \item If $a_l$ sets $a_l.P=1$, then $a_i$ also sets $a_i.P=1$.
\end{itemize}

\item \textbf{If $\bm{a_i.P=1}$:} Suppose agent $a_i$ at node $v$ receives $S_{v_1}$, $S_{v_2}$, \ldots, $S_{v_y}$ for nodes $v_1$, $v_2$, \ldots, $v_y$ respectively. Let $V_1=\cup_{i=1}^y \{v_i\} \cup \{v\}$. Let $k_1 = \sum_{u \in V_1} \alpha(u)$, where $\alpha(u)$ is extracted using $S_u$, $n_1$ be the cardinality of set $V_1$, $L=\big\lceil \frac{k_1}{n_1}\big \rceil$, and $l=\big\lfloor \frac{k_1}{n_1}\big \rfloor$.

Define \( x = \max\{\alpha(v) : v \in V_1\} \) and \( y = \min\{\alpha(v) : v \in V_1\} \). Let \( U_x = \{u_j \in V : \alpha(u_j) = x \text{ and }j\in[1,l]\} \), \( W_y = \{w_{j'} \in V : \alpha(w_{j'}) = y \text{ and }j'\in[1,l']\} \). For each \( u_j \in U_x \), let \( c_{j} \) be the minimum ID agent at node \( u_j \); for each \( w_{j'} \in W_y \), let \( b_{j'} \) be the minimum ID agent at node \( w_{j'} \). Without loss of generality, let $c_1$ be the minimum ID agent among $c_j$s, and $b_1$ be the minimum ID agent among $b_{j'}$s.

\noindent Agent \( a_i \) finds sets \( U_x \), \( W_y \), and the IDs of all \( a_j \) at node $u_j$ and \( b_{j'} \) at node $w_{j'}$. Also, using $Algo\_weak\_disp(k)$, it can construct a partial map of $\mathcal{G}_r$ (say $G'$). Using this information, it does the following at round $r$.

\noindent If agent $a_i$ is minimum ID agent at node $v$, then it finds a shortest path $P$ between $u_1$ and $w_1$ in $G'$. If there is more than one such shortest path, then it selects the lexicographically shortest path between $u_1$ and $w_1$. Let $P=z_1(=u_1)\sim z_2 \ldots \sim z_y(=w_1)$.
\begin{itemize}
    \item If $x\geq L+1$ or $y\leq l-1$,
    \begin{itemize}
        \item If $c_1.ID=a_i.ID$, then agent $a_i$ moves to node $z_2$. 
        \item If $c_1.ID\neq a_i.ID$, and it is at node $z_j$, and $j<{y}$, then it moves to node $z_{j+1}$. If agent $a_i$ is at node $z_{y}$, then it does not do anything.
    \end{itemize}
    \item Else if $x\leq L$ and $y\geq l$, then it terminates.
\end{itemize}

\noindent If \( a_i \) is not the minimum ID agent at some node \( u \), then it makes the following decision. Let $a_l$ be the minimum ID agent at node $u$. Since agent $a_i$ has the same information as agent $a_l$, it can calculate the decision of agent $a_l$ at round $r$. If it finds that agent $a_l$ decides to terminate in round $r$, then agent $a_i$ terminates. Otherwise, it does not do anything at round $r$.
\end{itemize}

\begin{theorem}\label{thm:algo_final_L_bdd}
    Given $k \geq 1$ agents placed arbitrarily on the nodes of $G$. Our algorithm solves $k$-balanced dispersion in $O(k)$ rounds with $O(\log n)$ memory at each agent in the synchronous setting with global communication and 1-hop visibility. Also, all agents understand that dispersion has been achieved and terminated. 
\end{theorem}
\begin{proof}
If there is a hole and a multinode at round $r$, every agent receives this information. Due to Lemma~\ref{lm:decrement_holes}, $Algo\_weak\_disp(k)$ decreases the number of holes by at least one in every round. If $k \leq n$, then by Theorem~\ref{thm:main_result}, agents achieve dispersion within the first $n$ rounds and terminate. If $k > n$, initially there can be at most $n-1$ holes, so by Lemma~\ref{lm:decrement_holes}, there will be no hole after the first $n$ rounds. After that, each agent $a_i$ sets $a_i.P = 1$. Therefore, at any round $r$ where $a_i.P = 1$, every node in $\mathcal{G}_r$ has at least one agent. Since each node is occupied, $Algo\_weak\_disp(k)$ enables every agent to construct the map of $\mathcal{G}_r$. Therefore, in our algorithm at round $r$, graph $G'$ is nothing but $\mathcal{G}_r$. Since each node has at least one agent, set $V_1$ is nothing but $V$, $k_1 = k$ and $n_1 = n$. Hence, agents correctly know that at most \( L = \lceil k/n \rceil \) agents and at least \( l = \lfloor k/n \rfloor \) agents can be present at a node. Once \( a_i.P = 1 \), each agent \( a_i \) can compute the same values of \( x \), \( y \), \( U_x \), and \( U_y \). Therefore, for every agent, the choice of the path \( P \) between nodes \( u_1 \) and \( w_1 \) is unique in our algorithm. In this case, the size of $u_1$ decreases by 1 at the end of round $r$. Agents continue this process till they find that every node has at least $l$ agents and at most $L$ agents. Therefore, when agents terminate, agents are in a $k$-balanced dispersed configuration.

 For $k \leq n$, dispersion is achieved in $O(k)$ rounds as per Theorem~\ref{thm:main_result}. For $k > n$, dispersion completes in $O(n + k)=O(k)$ rounds, where $O(n)$ to remove all holes, then $O(k)$ to maintain agents between $l$ and $L$ agents at every node, as only one node’s number of agents reduces per round. Agents only remember $a_i.P$ and $a_i.ID$, which requires $O(\log n)$ memory per agent. This completes the proof.
\end{proof}

\begin{remark}
    If $n+1 \leq k \leq n^{c_1}$ for some constant $c_1$, then by Lemma~\ref{lm:memory_lower}, the memory required by the agents is $\Theta(\log n)$. Hence, in this range of $k$, the memory requirement by the agents is optimal.
\end{remark} 

\section{Rooted $k$-balanced dispersion in 1-bounded 1-interval connected graphs}\label{sec:1-bdd-1-int}
In Section~\ref{sec:imp_l-bdd}, we have shown that solving $k$-balanced dispersion for $k = pn$, $p \in \mathbb{N}$, in 1-bounded 1-interval connected graphs requires both 1-hop visibility and 1-hop communication. In this section, we present two algorithms to solve $k$-balanced dispersion in 1-bounded 1-interval connected graphs when $k$ agents are co-located. The first algorithm is when $k = n + 1$, and the second extends this to $k \in \{pn + 1, pn + 2\}$, where $p \in \mathbb{N}$. The question remains open whether 1-hop visibility and 1-hop communication are necessary to solve $k$-balanced dispersion for $k = pn + q$, where $q \in [3, n - 1]$, in 1-bounded 1-interval connected graphs. We begin by describing the algorithm for $k = n + 1$, and then generalise it to the case of $k = pn + 1$ and $k = pn + 2$. Both algorithms assume agents have face-to-face (f-2-f) communication, 0-hop visibility, and prior knowledge of $n$.

\subsection{Dispersion with $n+1$ co-located agents}\label{sec:gen_graph}
  In the algorithm, we use the idea of a depth-first search (DFS) traversal by mobile agents from \cite{Augustine_2018}, which also solves the dispersion problem in static graphs $G$ for $k\leq n$ agents. Before describing the main idea of \cite{Augustine_2018}, we first introduce three technical terms used in their work, which we also use in our algorithm description.

\begin{itemize}
    \item \textbf{Settled agent:} An agent $a_i$ is said to be settled at node $v$ if it occupies $v$ at round $r$ and never leaves $v$ in any round after $r$.
    \item \textbf{Unsettled agent:} An agent is unsettled if it is not settled.
    \item \textbf{Root node of the DFS:} The initial node from which an agent starts the DFS traversal is called the root node of the DFS.
\end{itemize}

The idea of \cite{Augustine_2018} to solve dispersion for $k\leq n$ co-located agents is as follows. At each node $v$, a group of agents chooses an unexplored port $p$ and traverses through it: (i) if it reaches a node $u$ and finds that there is no agent settled at $u$, then one agent settles at $u$, and stores port $p'$ (through which it entered node $u$) as its parent port. Now, the unsettled agents try to move through port $(p'+1)$ mod $\delta_u$ in the explore state. (ii) Else, if the group of agents reach a node which already has a settled agent, the agent returns to node $v$ in the backtrack state. If the unsettled agents reach node $v$ in the backtrack state, then the agents compute $q=(p+1)$ mod $\delta_v$. If port $q$ is unexplored, it moves through port $q$ in the explore state. Else, if port $q$ equals the parent port of $v$, then the agents move through port $q$ in the backtrack state. The DFS algorithm takes $4m$ rounds to run on any graph, where $m$ is the number of edges in that graph. The information required for the group of agents at each node to run this algorithm can be stored by the settled agent that is already present at that node.

Recall that, according to our model, an agent does not know if any associated edge with respect to its current position is missing or not. This makes the adversary powerful as it can precompute the agents' decision at a node and accordingly can delete an associated edge at the start of that round. If the group of unsettled agents decides to move through port $p$ according to our algorithm in round $r$, the adversary can remove the edge corresponding to port $p$ in round $r$. Thus, preventing the group of agents from moving in round $r$. 

\vspace{0.3cm}
\noindent \textbf{High-level idea:} No group of agents moves indefinitely through which its movement was unsuccessful. Initially, there is only one group as the configuration is rooted. Whenever the group has an unsuccessful move, the agents divide into two groups, namely $G_1$ and $G_2$. These two groups run their dispersion algorithm via DFS separately. A missing edge can make the movement unsuccessful for both groups. On unsuccessful movement, $G_1$ never deviates from its DFS traversal path, whereas $G_2$ deviates after some finite waiting period. This keeps the algorithm running, and either $G_2$ gets dispersed, or $G_1$ can move one hop according to its DFS traversal. If $G_2$ gets dispersed, $G_1$ can understand it within some finite rounds and get divided into two groups further. Else, if $G_1$ gets dispersed by moving one hop again and again, $G_2$ can understand it within some finite rounds and gets divided into two groups further. We ensure that none of $G_1$ (or $G_2$) divides further unless $G_2$ (or $G_1$) is already dispersed. This keeps the total number of groups bounded by 2. As the settled agents keep the information required for the DFS traversal of all the groups, it helps the agents store the information using $O(\log n)$ memory. As we have different algorithms for the groups and different criteria for the groups to be divided, while dividing a group into two, each new group is given a unique label so that $G_1\text{ and }G_2$ can execute their respective algorithms.
\begin{algorithm}
    \caption{Dispersion in $1$-TVG with $n+1$ agents}\label{algo:disp}
    \While{True}
    {
        $a_i.r=a_i.r+1$\\
        \If{$a_i$ belongs to $G_1$ and $a_i.settled=0$}
        {
            call Algorithm \ref{algo:grp_1}
        }
        \ElseIf{$a_i$ belongs to $G_2$ and $a_i.settled=0$}
        {
            call Algorithm \ref{algo:grp_2}
        }
        \ElseIf{$a_i.settled=1$, there is $a_j$ from $G_1$ at the same node and $a_j.settled=0$}
        {
            call Algorithm \ref{algo:settled-1}
        }
        \ElseIf{$a_i.settled=1$, there is $a_j$ from $G_2$ at the same node and $a_j.settled=0$}
        {
            call Algorithm \ref{algo:settled-2}
        }
    }
\end{algorithm}
\subsubsection{The algorithm}\label{sec:algorithm}
In this section, we provide a detailed description of the algorithm for dispersion in 1-bounded 1-interval connected graph $\mathcal{G}$. In our algorithm, unsettled agents may get divided into two groups, which we call $G_1$ and $G_2$ respectively. At any particular round, there can be no more than two groups. However, over time, if all the unsettled agents of one group get settled, the other group can be divided into two groups again. So, if we look throughout the run of the algorithm, there may exist several groups. To identify these groups, we provide unique labelling to each group whenever two groups are formed, and they store their new ID in a parameter. Now we provide the parameters maintained by each agent $a_i$.

% Movement of Group$_1$

\begin{algorithm}
    \caption{Movement of agent $a_i \in G_1$ at $v$ with $a_i.settled=0$}\label{algo:grp_1}
    \If{$a_i.state=explore$}
    {
        \If{$r$ mod 2 =1}
        {
            \If{$a_i.success=True$}
            {
                set $a_i.count_1=0$\\
                \If{there is no agent at node $v$ with $settled=1$}
                {
                    \If{$a_i.ID$ is the minimum from all the unsettled agents present at the current node}
                    {
                        set $a_i.settled=1$\\
                        update $a_i^v(G_1).(parent, \, label, \, dfs\_label)=(a_i.prt\_in,\, a_i.grp\_label,\, a_i.dfs\_label)$\\
                        \If{there is an agent $a_j$ from $G_2$ with $a_j.settled=0$}
                        {
                            update $a_i^v(G_2).(parent,\, label, \, dfs\_label)=(a_j.prt\_in, \, a_j.grp\_label, \, a_j.dfs\_label)$
                        }
                    }
                    \Else
                    {
                        set $a_i.prt\_out=(a_i.prt\_in+1)\mod\delta_v$\\
                        \If{$a_i.prt\_out=a_i.prt\_in$}
                        {
                            set $a_i.state=backtrack$ and move through $a_i.prt\_out$
                                
                        }
                        \Else
                        {
                            move through $a_i.prt\_out$
                        }
                    }
                }
                \ElseIf{there is an agent $r_j$ at node $v$ with $settled=1$}
                {
                    $a_i.prt\_out=(a_i.prt\_in+1)$ mod $\delta_v$\\
                    \If{$label=a_i.grp\_label$, where $label$ is the component of $r_j^v(G_1)$}
                    {
                        $a_i.state=backtrack$ and move through $a_i.prt\_in$
                    }
                    \ElseIf{$label \neq a_i.grp\_label$, where $label$ is the component of $r_j^v(G_1)$}
                    {
                        move through $a_i.prt\_out$
                    }
                }               
            }
            \ElseIf{$a_i.success=False$}
            {
                \If{$a_i.divide=0$}
                {
                    set $a_i.divide=1$ and call $Group\_divide()$
                }
                \ElseIf{$a_i.divide=1$}
                {
                    update $a_i.count_1=a_i.count_1+1$\\
                    \If{$a_i.count_1=16n^2$}
                    {
                        call $Group\_divide()$
                    }
                    \ElseIf{$a_i.count_1<16n^2$}
                    {
                        move through $a_i.prt\_out$
                    }
                }
            }
    }
     \ElseIf{$r$ mod 2=0}
        {
            \If{the movement of agent $a_i$ in round $r-1$ is successful}
            {
                set $a_i.success=True$
            }
            \Else
            {
                set $a_i.success=False$
            }
        }
    }
    \ElseIf{$a_i.state=backtrack$}
    {
        \If{$r$ mod 2=1}
        {
            \If{$a_i.success=True$}
            {
                set $a_i.count_1=0$\\
                \If{$parent$ component of $r_j^v(G_1)$ is $-1$, where $r_j$ is the settled agent at node $v$}
                {
                    set $a_i.prt\_out=(a_i.prt\_in+1)\mod\delta_v$\\
                    set $a_i.state=explore$ and move through $a_i.prt\_out$
                }
                \ElseIf{$parent$ component of $r_j^v(G_1)$ is not $-1$, where $r_j$ is the settled agent at node $v$}
                {
                    
                    set $a_i.prt\_out=(a_i.prt\_in+1)\mod\delta_v$\\
                    \If{$a_i.prt\_out=parent$ and $parent$ is the component of $r_j^v(G_1)$, where $r_j$ is the settled agent at node $v$}
                    { 
                        move through $a_i.prt\_out$
                    }
                    \Else
                    {
                        set $a_i.state=explore$ and move through $a_i.prt\_out$
                    }
                }
            }
            \ElseIf{$a_i.success=False$}
            {
                \If{$a_i.divide=0$}
                {
                    set $a_i.divide=1$ and call $Group\_divide()$
                }
                \ElseIf{$a_i.divide=1$}
                {
                    update $a_i.count_1=a_i.count_1+1$\\
                    \If{$a_i.count_1=16n^2$}
                    {
                        call $Group\_divide()$
                    }
                    \ElseIf{$a_i.count_1<16n^2$}
                    {
                        move through $a_i.prt\_out$
                    }
                }
            }
        }
        \If{$r$ mod 2=0}
        {
            \If{the movement of agent $a_i$ in round $r-1$ is successful}
            {
                set $a_i.success=True$
            }
            \Else
            {
                set $a_i.success=False$
            }
        }
    }
\end{algorithm}

%%%%%%%%%%%%%%%%%%%%%%%%%%%%%%%%%%%%%%%%%%%%%%%%%%%%%%%%%%%%%%%%
\begin{algorithm}
    \caption{$Group\_divide()$}\label{alg:grp_divide}
    let $\{a_1, a_2, ..., a_x\}$ be the unsettled agents present at the current node in the increasing order of their IDs \\
    \If{$i\leq \lceil{\frac{x}{2}}\rceil$}
    {
        update $a_i.grp\_label=a_i.grp\_label.0$ \\
        $a_i.prt\_out=0$\\
        set $a_i.state=explore$\\
        $a_i.skip=-1$\\
        set $a_i.dfs\_label=a_i.dfs\_label+1$\\
        move through $a_i.prt\_out$ 
    }  
    \Else
    {
        update $a_i.grp\_label=a_i.grp\_label.1$ \\
        $a_i.prt\_out=0$
        set $a_i.state=explore$\\
        $a_i.skip=-1$\\
        set $a_i.dfs\_label=a_i.dfs\_label+1$\\
        move through $a_i.prt\_out$
    }
\end{algorithm}

\begin{itemize}
    \item $\bm{a_i.settled}$: It is a binary variable and indicates whether $a_i$ is settled at the current node or not.  If $a_i.settled=0$ then it is not settled whereas $a_i.settled=1$ denotes that agent $a_i$ is settled. Initially $a_i.settled=0$.
    
    \item $\bm{a_i.prt\_in}$: It stores the port through which an agent $a_i$ enters into the current node. Initially, $a_i.prt\_in=-1$. 
    
    \item $\bm{a_i.prt\_out}$: It stores the port through which an agent $a_i$ exits from the current node. Initially, $a_i.prt\_out=-1$.
    
    \item $\bm{a_i.state}$: It denotes the state of an agent $a_i$ that can be either explore or backtrack. Initially $a_i.state=explore$.
    
    \item $\bm{a_i.dfs\_label}$: It stores the number of new DFS traversals started by agent $a_i$. Initially $a_i.dfs\_label=1$.

    \item $\bm{a_i.success}$: When an agent $a_i$ attempts to move from a node $u$ to $v$ via the edge $(u, \,v)$, in the round $t$, it is either successful or unsuccessful. If the agent reaches the node $v$ in the round $t+1$, it is a successful attempt, thus, $a_i.success=True$; otherwise, $a_i.success=False$. Each agent $a_i$ initially has $a_i.success=True$.

    \item $\bm{a_i.skip}$: If agent $a_i$ skips a port, then it stores it in this parameter. Initially, $a_i.skip=-1$. This parameter allows agents of $G_2$ to store port numbers when they skip at the root of their DFS algorithm.

    \item $\bm{a_i.divide}$: Initially, all agents are together and each agent has $a_i.divide=0$. If and when this group of agents divides into two groups for the first time, each agent sets $a_i.divide=1$. There is no change in the values of this variable thereafter.
    
    \item $\bm{a_i.grp\_label}$: It stores a binary string. Initially $a_i.grp\_label=10$. If $a_i.grp\_label$ ends with 0, it implies $a_i$ is a part of $G_1$. Else if $a_i.grp\_label$ ends with 1, it implies $a_i$ is a part of $G_2$. Later, we elaborate on how this variable gets updated. 
    
    \item $\bm{a_i.count_1}$: If agent $a_i \in G_1$, then it stores the number of consecutive odd rounds with $a_i.success=False$. Initially, each agent has $a_i.count_1 = 0$. This parameter is used by agents in $G_1$ to determine when $G_2$ has been fully dispersed and when $G_1$ can be further divided into two new groups.

    \item $\bm{a_i.count_2}$: If agent $a_i \in G_2$ then it stores the number of consecutive odd rounds with $a_i.success=False$. This parameter is used by agents in $G_2$ to determine when $G_2$ has needs

    \item $\bm{a_i.count_3}$: If agent $a_i \in G_2$, then it stores the number of new DFS traversals it started from the time $a_i$ became a part of this group. This parameter is used by agents in $G_2$ to determine when $G_1$ has been fully dispersed and when $G_2$ can be further divided into two new groups.
\end{itemize}

    Let $a_i$ be a settled agent at node $v$. Any unsettled agent belongs to either $G_1$ or $G_2$. When an unsettled agent $a_j$ while doing its movement (along with other unsettled agents, if any, in its group) according to our algorithm, reaches a node $v$ where an agent $a_i$ is already settled, $a_j$ needs some information regarding $v$ from $a_i$ to continue its traversal. In this regard, each settled agent maintains two parameters, one with respect to each of the groups as mentioned below.

\begin{itemize}
    \item $\bm{a_i^{v}(G_1).(parent, \, label, \, dfs\_label)}$: Initially $a_i^{v}(G_1).(parent, \, label, \,  dfs\_label)=(-1, \,-1, \,0 )$. Let $t$ be the round when agent $a_i$ settles at node $v$. Let $t'\geq t$ be the last round when $a_i$ saw some agent of $G_1$, say, agent $a_j$ at node $v$.  In $a_i^{v}(G_1).(parent, \, label, \, dfs\_label)$, $a_i$ stores $(a_j.prt\_in, \,a_j.grp\_label,\, a_j.dfs\_label)$. That is, it stores the information regarding node $v$ w.r.t. the DFS traversal of $G_1$ in round $t'$. Note that throughout this paper, we denote $a_i^{v}(G_1).(parent, \, label, \, dfs\_label)$ by $a_i^{v}(G_1)$.

    \vspace{0.1cm}
    \item $\bm{a_i^{v}(G_2).(parent, \, label, \, dfs\_label)}$: Initially $a_i^{v}(G_2).(parent, \, label, \,  dfs\_label)=(-1, \,-1, \,0 )$. Let $t$ be the round when agent $a_i$ settles at node $v$. Let $t'\geq t$ be the last round when $a_i$ saw some agent of $G_2$, say, agent $a_j$ at node $v$.  In $a_i^{v}(G_2).(parent, \, label, \, dfs\_label)$, $a_i$ stores $(a_j.prt\_in, \,a_j.grp\_label,\, a_j.dfs\_label)$. That is, it stores the information regarding node $v$ w.r.t. the DFS traversal of $G_2$ in round $t'$. Note that throughout this paper, we denote $a_i^{v}(G_2).(parent, \, label, \, dfs\_label)$ by $a_i^{v}(G_2)$. 
\end{itemize}

  Suppose $n+1$ agents are co-located at a node, say $v$, of the graph (rooted configuration). Initially, for each agent $a_i$, $a_i.grp\_label=10$. Therefore, all agents are a part of $G_1$. According to our model, an agent cannot sense the presence of a missing edge unless it attempts to move through it. To tackle this, we take two consecutive rounds to perform one edge traversal. In every odd round, agents move, and in the next even round, it realizes the movement is successful or unsuccessful. Precisely, say an agent $a_i$ attempts to move in an odd round $r$. If its movement is successful, it updates its parameter $a_i.success$ to $True$ in the even round $r+1$. Otherwise, it updates $a_i.success=False$ in the even round $r+1$. Now we proceed with a detailed description of our algorithm (Refer to Algorithm \ref{algo:disp} for the main algorithm).

\vspace{0.1cm}
\noindent \textbf{Algorithm for unsettled agent $\bm{a_i}$ of $\bm{G_1}$ at odd round:} In any odd round, let agent $a_i\in G_1$ be at node $u$ and $a_i.settled=0$. Based on $a_i.divide$, it follows the following steps.

\begin{itemize}
    \item $\bm{a_i.divide=0}$: Suppose, agent $a_i$ reaches a node $w$ and finds no agents $a_j$ with $a_j.settled=1$ (i.e., no settled agents at node $w$). In this case, it follows the following steps. If agent $a_i$ is the minimum ID at node $w$, then it settles at $w$ and updates $a_i.settled=1$, $a_i.grp\_label=\bot$ (lines 5-8 of Algorithm \ref{algo:grp_1}). Further, it updates $a_i^w(G_1)$. Else if agent $a_i$ is not the minimum ID, then it moves according to the DFS traversal algorithm (lines 11-17 of Algorithm \ref{algo:grp_1}). If agent $a_i$ reaches a node $v$ and finds an agent with $a_j$, with $a_j.settled=1$, then $a_i$ does not settle at node $v$ and moves according to the DFS traversal algorithm (lines 18-24 and lines 42-54 of Algorithm \ref{algo:grp_1}). In this way, the agent $a_i \in G_1$ continues the DFS traversal of the graph unless $a_i$ encounters a missing edge for the first time. In particular, at some odd round $r$, agent $a_i$ moves via an edge $e$, and the movement is not successful at the end of round $r$. This is understood by agent $a_i$ in the round $r+1$. In round $r+2$, agent $a_i$ updates $divide$ by 1 (lines 23-25 and lines 50-52 of Algorithm \ref{algo:grp_1}) and divides into two groups using the $Group\_divide$ procedure as follows.

\vspace{0.15cm}
    \noindent 
    $\bm{Group\_divide}$: (refer to Algorithm \ref{alg:grp_divide}) Half of the agents with smaller IDs at node $u$ are part of $G_1$ and the other half of the agents at node $u$ are part of $G_2$. If agent $a_i$ is part of $G_2$, then sets $a_i.grp\_label=a_i.grp\_label.1$ by appending 1 at the end and starts a new DFS traversal considering the current node as the new root. In this case, agent $a_i$ also updates the other parameters $a_i.skip=-1$, $a_i.state=explore$, $a_i.dfs\_label=a_i.dfs\_label+1$ and $a_i.prt\_out=0$. Else if agent $a_i$ remains as a part of $G_1$, then sets $a_i.grp\_label=a_i.grp\_label.0$ by appending 0 at the end and starts a new DFS traversal considering the current node as the new root. In this case, agent $a_i$ also updates the other parameters $a_i.skip=-1$, $a_i.state=explore$, $a_i.dfs\_label=a_i.dfs\_label+1$ and $a_i.prt\_out=0$. Agent $a_i$ moves through $a_i.prt\_out$.

    \vspace{0.2cm}
    \item $\bm{a_i.divide=1}$: (refer to Algorithm \ref{algo:grp_1}) Suppose, agent $a_i \in G_1$ reaches a node $w$ and finds no agents $a_j$ with $a_j.settled=1$ (i.e., no settled agents at node $w$). In this case, it follows the following steps. If agent $a_i$ is the minimum ID at node $w$, then it settles at $w$ and updates $a_i.settled=1$, $a_i.grp\_label=\bot$. Further, it updates $a_i^w(G_1) =(a_i.prt\_in, \,a_i.grp\_label, \, a_i.dfs\_label)$. If there is any unsettled agent $a_j$ from $G_2$ at node $w$, then $a_i$ updates $a_i^u(G_2)$ to $(a_j.prt\_in, a_j.grp\_label, a_j.dfs\_label)$. If $a_i\in G_1$ reaches a node $w$ and find agent $a_j$ with $a_j.settled=1$, then it does not settle at the node $v$ and it can understand whether node $w$ is the new node according to the current DFS or previously visited node by comparing $label$, $dfs\_label$ components of $a_j^w(G_1)$ with $a_i.grp\_label$, $a_i.dfs\_label$ respectively. If both matches then node $w$ is visited previously in the current DFS traversal of agent $a_i$. Otherwise, if any of those two components do not match, then agent $a_i\in G_1$ treats $w$ as an unvisited node w.r.t. the current DFS. Round $r+2$ onwards, $a_i$ updates its $count_1$ parameter as follows. At odd round $r_1\geq r+2$, if it finds $a_i.success=True$, i.e. it realized a successful movement in the even round $r_1-1$, agent $a_i$ updates $a_i.count_1=0$; else if $a_i.success=False$, agent $a_i$ updates $a_i.count_1=a_i.count_1+1$. If at some odd round, $a_i$ finds $a_i.count_1=16n^2$, then unsettled agents at node $u$ divide into two groups again using aforementioned $Group\_divide$ procedure.\footnote{ Note that all unsettled agents at node $u$ belong are from \( G_1 \). We will show in the analysis that once the \( count_1 \) parameter for agents in \( G_1 \) reaches \( 16n^2 \), all agents in \( G_2 \) becomes settled agent.} That is, $G_1$ divides into two groups only after getting stuck at a node for a particular port for $16n^2$ many consecutive odd rounds. 
    \end{itemize}

\begin{algorithm}
    \caption{$a_i.state=backtrack$}\label{a_i.state:backtrack}
    \If{$r$ mod 2 =1}
    {
        \If{$a_i.success=True$}
        {
            set $a_i.prt\_out=(a_i.prt\_in+1)\mod\delta_v$ and $a_i.count_2=0$\\
            \If{ $parent$ component of $r_j^v(G_2)$ is $-1$, where $r_j$ is the settled agent at node $v$}
            {
                \If{$a_i.prt\_out = $  minimum available port except $a_i.skip$}
                {
                    $a_i.state=explore$, $a_i.skip=-1$, $a_i.prt\_out=0$, and $a_i.dfs\_label=a_i.dfs\_label+1$\\
                    move through $a_i.prt\_out$
                }
                \ElseIf{$a_i.prt\_out \neq $ minimum available port except $a_i.skip$}
                {
                    \If{$a_i.prt\_out=a_i.skip$}
                    {
                        $a_i.prt\_out=(a_i.prt\_out+1)\mod\delta_v$\\
                        \If{$a_i.prt\_out = $  minimum available port except $a_i.skip$}
                        {
                            $a_i.state=explore$, $a_i.skip=-1$, $a_i.prt\_out=0$, and $a_i.dfs\_label=a_i.dfs\_label+1$\\
                            move through $a_i.prt\_out$
                        }
                        \ElseIf{$a_i.prt\_out \neq  $  minimum available port except $a_i.skip$}
                        {
                            $a_i.state=explore$ and move through $a_i.prt\_out$
                        }
                    }
                    \ElseIf{$a_i.prt\_out\neq a_i.skip$}
                    {
                        $a_i.state=explore$ and move through $a_i.prt\_out$
                    }
                }
            }
            \ElseIf{$parent$ component of $r_j^v(G_2)$ is not $-1$, where $r_j$ is the settled agent at node $v$}
            {
                \If{$a_i.prt\_out=parent$, where $parent$ is the component of $r_j^v(G_2)$}
                {
                    move through $a_i.prt\_out$
                }
                \ElseIf{$a_i.prt\_out \neq parent$, where $parent$ is the component of $r_j^v(G_2)$}
                {
                    $a_i.state=explore$ and move through $a_i.prt\_out$
                }
            }
        }
        \ElseIf{$a_i.success=False$}
        {
            
            \If{$b_i$ from $G_1$ is present at node $v$ and $b_i.prt\_out=a_i.prt\_out$}
            {
                $a_i.count_2=0$, $a_i.dfs\_label=a_i.dfs\_label+1$, $a_i.state=explore$, and $a_i.skip=a_i.prt\_out$\\
                set $a_i.prt\_out$ as the minimum port available at the current node except $a_i.skip$\\
                move through $a_i.prt\_out$
            }
            \ElseIf{$b_i$ from $G_1$ is present at node $v$ and $b_i.prt\_out \neq a_i.prt\_out$ or there is no agent from $G_1$ present at node $v$}
            {
                $a_i.count_2=a_i.count_2+1$\\
                \If{$a_i.count_2=4n^2$}
                {
                    set $a_i.count_2=0$, and $a_i.count_3=a_i.count_3+1$\\
                    \If{$a_i.count_3<4n^2$}
                    {
                        set $a_i.state=explore$, $a_i.dfs\_label=a_i.dfs\_label+1$, and $a_i.skip=a_i.prt\_out$\\
                        set $a_i.prt\_out$ as the minimum port available at the current node except $a_i.prt\_out$\\
                        move through $a_i.prt\_out$
                    }
                    \Else
                    {
                        call $Group\_divide()$
                    }
                }
                \ElseIf{$a_i.count_2<4n^2$}
                {
                    $a_i.count_2=a_i.count_2+1$\\
                    move through $a_i.prt\_out$
                }
            }
        }
    }
    \ElseIf{$r$ mod 2 = 0}
    {
        \If{the movement of agent $a_i$ in round $r-1$ is successful}
        {
            set $a_i.success=True$
        }
        \Else
        {
            set $a_i.success=False$
        }
    }
    
\end{algorithm}

\begin{algorithm}
    \caption{$a_i.state=explore$}\label{a_i.state:explore}
    \If{$r$ mod 2 =1}
        {
            \If{$a_i.success=True$}
            {
                set $a_i.count_2=0$\\
                \If{there is no agent at node $v$ with $settled=1$}
                {
                    \If{$a_i.ID$ is the minimum from all the unsettled agents present at node $v$}
                    {
                        set $a_i.settled=1$\\
                        update $a_i^v(G_2).(parent, \, label, \, dfs\_label)=(a_i.prt\_in, \,a_i.grp\_label, \, a_i.dfs\_label)$\\
                        \If{there is an agent $a_j$ from $G_1$ with $a_j.settled=0$ at node $v$}
                        {
                            update $a_i^v(G_1).(parent, \,label, \, dfs\_label)=(a_j.prt\_in, \, a_j.grp\_label, \, a_j.dfs\_label)$
                        }
                    }
                    \Else
                    {
                        set $a_i.prt\_out=(a_i.prt\_in+1)\mod\delta_v$\\
                        \If{$a_i.prt\_out=a_i.prt\_in$}
                        {
                            set $a_i.state=backtrack$
                        }
                        move through $a_i.prt\_out$
                    }
                }
                \ElseIf{there is an agent $r_j$ at node $v$ with $r_j.settled=1$}
                {
                    \If{$a_i.grp\_label= label$ and $a_i.dfs\_label=dfs\_label$, where $label$ and $dfs\_label$ are the component from $r_j^v(G_2)$}
                    {
                        set $a_i.state=backtrack$\\
                        set $a_i.prt\_out=a_i.prt\_in$\\
                        move through $a_i.prt\_out$
                    }
                    \ElseIf{$a_i.grp\_label\neq label$ or $a_i.dfs\_label\neq dfs\_label$, where $label$ and $dfs\_label$ are the component from $r_j^v(G_2)$}
                    {
                        set $a_i.prt\_out=(a_i.prt\_in+1)\mod\delta_v$\\
                        \If{$a_i.prt\_out=a_i.prt\_in$}
                        {
                            set $a_i.state=backtrack$\\
                            move through $a_i.prt\_out$
                        }
                        \Else
                        {
                            move through $a_i.prt\_out$
                        }
                    }
                }
            }
            \ElseIf{$a_i.success=False$}
            {
                \If{ $b_i$ from $G_1$ is present at node $v$ and $b_i.prt\_out=a_i.prt\_out$}
                {
                    \If{$parent$ component of $r_j^v(G_2)$ is $-1$, where $r_j$ is the settled agent at node $v$}
                    {
                        \If{$a_i.prt\_out=\delta_v -1$}
                        {
                            $a_i.count_2=0$, $a_i.skip=-1$, $a_i.dfs\_label=a_i.dfs\_label+1$, $a_i.prt\_out=0$\\
                            move through $a_i.prt\_out$
                        }
                        \ElseIf{$a_i.prt\_out \neq \delta_v -1$}
                        {
                            $a_i.count_2=0$, $a_i.prt\_out=(a_i.prt\_out+1)$ mod $\delta_v$\\
                            move through $a_i.prt\_out$
                        }
                    }
                    \ElseIf{$parent$ component of $r_j^v(G_2)$ is not $-1$, where $r_j$ is the settled agent at node $v$}
                    {
                        $a_i.count_2=0$\\
                        $a_i.prt\_out=(a_i.prt\_out+1)$ mod $\delta_v$\\
                        \If{$a_i.prt\_out=a_i.prt\_in$}
                        {
                            $a_i.state=backtrack$ and move through $a_i.prt\_out$
                        }
                        \ElseIf{$a_i.prt\_out\neq a_i.prt\_in$}
                        {
                            move through $a_i.prt\_out$
                        }
                    }                    
                }
                \ElseIf{$b_i$ from $G_1$ is present at node $v$ and $b_i.prt\_out \neq a_i.prt\_out$ or there is no agent from $G_1$ present at node $v$}
                {
                    set $a_i.count_2=a_i.count_2+1$\\
                    \If{$a_i.count_2<4n^2$}
                    {
                        move through $a_i.prt\_out$
                    }
                    \ElseIf{$a_i.count_2=4n^2$}
                    {
                        $a_i.count_2=0$ and $a_i.count_3=a_i.count_3+1$\\
                        \If{$a_i.count_3<4n^2$}
                        {
                            set $a_i.state=explore$, $a_i.dfs\_label=a_i.dfs\_label+1$, $a_i.skip=a_i.prt\_out$\\
                            set $a_i.prt\_out$ as the minimum port available at the current node except $a_i.skip$\\
                            move through $a_i.prt\_out$
                        }
                        \Else
                        {
                            call $Group\_divide()$
                        }
                    }
                    
                }
            }
        }
        
        \ElseIf{$r$ mod 2=0}
        {
            \If{the movement of agent $a_i$ in round $r-1$ is successful}
            {
                set $a_i.success=True$
            }
            \Else
            {
                set $a_i.success=False$
            }
        }
\end{algorithm}
%%%%%%%%%%%%%%%%%%%%%%%%%%%%%%%%%%%%%%%%%%%%%%%%%%%%%%%%%%%%%%%%%%%%%%%%%

%%%%%%%%%%%%%%%%%%%%%%%%%%%%%%%%%%%%%%%%%%%%%%%%%%%%%%%%%%%%%%%%%%%%%%%%%%%%%%%%%%
%Algorithm for unsettled agent of group 2
\begin{algorithm}
    \caption{Movement of agent $a_i \in G_2$ at $v$ with $a_i.settled=0$}\label{algo:grp_2}
    \If{$a_.state=backtrack$}
    {
        call Algorithm \ref{a_i.state:backtrack}
    }
    \Else
    {
        call Algorithm \ref{a_i.state:explore}
    }
\end{algorithm}

%%%%%%%%%%%%%%%%%%%%%%%%%%%%%%%%%%%%%%%%%%%%%%%%%

%%%%%%%%%%%%%%%%%%%%%%%%%%%%%%%%%%%%%%%%%%%%%%%%%%%%%%%%%%%%%%%%%

%%%%%%%%%%%%%%%%%%%%%%%%%%%%%%%%%%%%%%%%%%%%%%%%%%%%%%%%%%%%%%%%%%%%%%%%%%%%%%%%%%%%%
% Settled agent algorithm for Grp2

%%%%%%%%%%%%%%%%%%%%%%%%%%%%%%%%%%%%%%%%%%%%%%

% \begin{algorithm}
%     \caption{Movement of agent $a_i \in G_2$ at $v$ with $a_i.settled=0$}\label{algo:grp_2}
%     \If{$a_.state=explore$}
%     {
%         call Algorithm \ref{a_i.state:explore}
%     }
%     \Else
%     {
%         call Algorithm \ref{a_i.state:backtrack}
%     }
% \end{algorithm}
\vspace{0.1cm}
\noindent \textbf{Algorithm for unsettled agent $\bm{a_i}$ of $\bm{G_2}$ at odd round:} In any odd round, let agent $a_i\in G_2$ be at node $u$ and $a_i.settled=0$. It follows the following steps.

\begin{itemize}
    \item[(1)] \textbf{If} $\bm{a_i.state=explore}$ \textbf{and} $\bm{a_i.success=True}$: It updates $a_i.count_2=0$ and does the following (lines 2-26 of Algorithm \ref{a_i.state:explore}).
    \begin{itemize}
        \item[(a)] \textbf{Node $\bm{u}$ has no agent with} $\bm{settled=1}$: If $a_i$ is the minimum ID at node $u$, it settles at node $u$ and updates $a_i.settled=1$, $a_i.grp\_label=\bot$ and $a_i^u(G_2)$ to $(a_i.prt\_in, a_i.grp\_label, a_i.dfs\_label)$. Note that in this case, there is no agent from \( G_1 \), since in the $Group\_divide$ procedure, the first half of the agents with smaller IDs are assigned to \( G_1 \).
 
        \vspace{0.1cm}
        
        \noindent If $a_i$ is not the minimum ID, then it updates $a_i.prt\_out=(a_i.prt\_in+1)$ mod $\delta_u$. If $a_i.prt\_out\neq a_i.prt\_in$, then it moves through $a_i.prt\_out$. Else, if $a_i.prt\_out= a_i.prt\_in$, then it sets $a_i.state=backtrack$ and move through $a_i.prt\_out$.  
        
        \item[(b)]  \textbf{Node $\bm{u}$ has an agent with} $\bm{settled=1}$: Let $a_j$ be at node $u$ with $a_j.settled=1$. If $a_i.grp\_label$, $a_i.dfs\_label$ matches $label$, $dfs\_label$ components of $a_j^u(G_2)$ respectively, then it is a previously visited node. In this case, $a_i$ sets $a_i.state=backtrack$, $a_i.prt\_out=a_i.prt\_in$, and it moves through $a_i.prt\_out$. Otherwise, it is an unvisited node w.r.t. the current DFS of $a_i$. In this case, it updates $a_i.prt\_out=(a_i.prt\_in+1)$ mod $\delta_u$. If $a_i.prt\_out\neq a_i.prt\_in$, then it moves through $a_i.prt\_out$. Else, if $a_i.prt\_out= a_i.prt\_in$, then it sets $a_i.state=backtrack$ and move through $a_i.prt\_out$. 
    \end{itemize}

    \item[(2)] \textbf{If} $\bm{a_i.state=explore}$ \textbf{and} $\bm{a_i.success=False}$: It does the following (lines 27-54 of Algorithm \ref{a_i.state:explore}).

    \begin{itemize}
        \item[(a)] \textbf{There is an agent from $\bm{G_1}$ with $\bm{settled=0}$ at node $\bm{u}$:} Let $a_j$ be settled at node $u$, $b_i \in G_1$ be at node $u$, and $b_i.settled=0$. In this case, there are two cases.

        \begin{itemize}
            \item[(i)] If $a_i.prt\_out=b_i.prt\_out$ (observe that based on communication, agent $a_i$ can compute what agent $b_i$ computes, and vice versa), then there are two cases possible. 
            
            \noindent ($\alpha$) \textbf{$\bm{parent}$ component of $\bm{a_j^u(G_2)}$ is $\bm{-1}$}: If $a_i.prt\_out=\delta_u-1$, then agent $a_i$ sets $a_i.state=explore$, stores $a_i.skip=-1$, $a_i.dfs\_label=a_i.dfs\_label+1$ (i.e., start a new DFS traversal), $a_i.prt\_out=0$, and moves through $a_i.prt\_out$. Else if $a_i.prt\_out \neq \delta_u-1$, then it sets $a_i.prt\_out=(a_i.prt\_out+1)$ mod $\delta_u$ and sets $a_i.count_2=0$ and move through $a_i.prt\_out$. 

            \noindent ($\beta$) \textbf{$\bm{parent}$ component of $\bm{a_j^u(G_2)}$ is not $\bm{-1}$}: It sets $a_i.prt\_out=(a_i.prt\_out+1)$ mod $\delta_u$ and sets $a_i.count_2=0$. If $a_i.prt\_out\neq a_i.prt\_in$, then it moves through $a_i.prt\_out$. Else, if $a_i.prt\_out= a_i.prt\_in$, then it sets $a_i.state=backtrack$ and move through $a_i.prt\_out$.

            % It updates $a_i.count_2=0$ and sets $a_i.skip=a_i.prt\_out$, $a_i.prt\_out=$ minimum available port except $a_i.skip$, $a_i.dfs\_label=a_i.dfs\_label+1$. Agent $a_i$ moves through $a_i.prt\_out$.
            \item[(ii)] If $a_i.prt\_out\neq b_i.prt\_out$, it updates $a_i.count_2=a_i.count_2+1$. If $a_i.count_2<4n^2$, then move through $a_i.prt\_out$. Else, if $a_i.count_2=4n^2$, then $a_i$ updates $a_i.count_2=0$ and $a_i.count_3=a_i.count_3+1$. In this case, agent $a_i$ sets $a_i.state=explore$, stores $a_i.skip=a_i.prt\_out$, $a_i.dfs\_label=a_i.dfs\_label+1$, $a_i.prt\_out=$ minimum available port except $a_i.skip$, and moves through $a_i.prt\_out$. Note that $a_i.count_3$ always remains less than $4n^2$ in this case as justified in the correctness.
        \end{itemize}
    
        \item[(b)] \textbf{There is no agent from $\bm{G_1}$ with $\bm{settled=0}$ at node $\bm{u}$:} It updates $a_i.count_2=a_i.count_2+1$. If $a_i.count_2<4n^2$, then move through $a_i.prt\_out$. Else, if $a_i.count_2=4n^2$, then $a_i$ updates $a_i.count_2=0$ and $a_i.count_3=a_i.count_3+1$. If $a_i.count_3<4n^2$, then agent $a_i$ sets $a_i.state=explore$, $a_i.dfs\_label=a_i.dfs\_label+1$, $a_i.skip=a_i.prt\_out$, $a_i.prt\_out=$ minimum available port except $a_i.skip$, and moves through $a_i.prt\_out$. Else, if $a_i.count_3=4n^2$, it divides into two groups using $Group\_divide$ procedure mentioned in the algorithm for unsettled agents of $G_1$. Note that if $a_i.count_3=4n^2$, then $G_1$ is already dispersed as justified in the correctness.
    \end{itemize}
    
     \item[(3)] \textbf{If} $\bm{a_i.state=backtrack}$ \textbf{and} $\bm{a_i.success=True}:$ It updates $a_i.prt\_out=(a_i.prt\_in+1)\mod \delta_u$ and $a_i.count_2=0$. Let $a_j$ be the agent at node $u$ with $a_j.settled=1$. Based on $parent$ component of $a_j^u(G_2)$, there are two cases (lines 2-22 of Algorithm \ref{a_i.state:backtrack}).
     \begin{itemize}
         \item[(a)] \textbf{$\bm{parent}$ component of $\bm{a_j^u(G_2)}$ is $\bm{-1}$}: If agent $a_i$ finds $parent$ component of $a_j^u(G_2)$ to be $-1$, then node $u$ is the root node of the current DFS traversal of agent $a_i$. In this case, if $a_i.prt\_out=$ minimum available port except for $a_i.skip$, then it sets $a_i.state=explore$, $a_i.skip=-1$, $a_i.prt\_out=0$, $a_i.dfs\_label=a_i.dfs\_label+1$, and moves through $a_i.prt\_out$. Else if $a_i.prt\_out \neq $ minimum available port except for $a_i.skip$, then it does the following. 
         
         If $a_i.prt\_out=a_i.skip$, then $a_i.prt\_out=(a_i.prt\_out+1)$ mod $\delta_v$. If $a_i.prt\_out=$ minimum available port except for $a_i.skip$, then it sets $a_i.state=explore$, $a_i.skip=-1$, $a_i.prt\_out=0$, $a_i.dfs\_label=a_i.dfs\_label+1$, and moves through $a_i.prt\_out$. Else if, it sets $a_i.state=explore$ and moves through $a_i.prt\_out$.

         If $a_i.prt\_out\neq a_i.skip$, it sets $a_i.state=explore$ and moves through $a_i.prt\_out$. 

         \item[(b)] \textbf{$\bm{parent}$ component of $\bm{a_j^u(G_2)}$ is not $\bm{-1}$}: If agent $a_i$ does not find $parent$ component of $a_j^u(G_2)$ to be $-1$, then node $u$ is not the root node of the current DFS traversal of agent $a_i$. If $a_i.prt\_out$ matches with $ parent$ component of $a_j^u(G_2)$, then agent $a_i$ moves through $a_i.prt\_out$. Else if $a_i.prt\_out$ does not match with $parent$ component of $a_j^u(G_2)$, then agent $a_i$ sets $a_i.state=explore$ and moves through $a_i.prt\_out$. 
     \end{itemize}
\begin{algorithm}
    \caption{Algorithm for settled agent $a_i$ at node $v$ and at least one unsettled agent $a_j$ from $G_1$ present at node $v$}\label{algo:settled-1}
    \If{$r$ mod 2=1}
    {
    \If{unsettled agent $a_j$ from $G_1$ and $a_j.success=True$}
    {
        \If{unsettled agents $a_j$ from $G_1$ is present at node $v$ with $a_j.state=backtrack$}
        {
            It does not do anything.
        }
        \If{unsettled agents $a_j$ from $G_1$ is present at node $v$ with $a_j.state=explore$}
        {
            \If{$label$ component of $a_i^v(G_1)$ does not match with $a_j.grp\_label$}
            {
                set $a_i^v(G_1).(parent, \,label, \,dfs\_label)=(a_j.prt\_in, \,a_j.grp\_label, \,a_j.dfs\_label)$  
            }
            \ElseIf{$label$ component of $a_i^v(G_1)$ matches with $a_j.grp\_label$}
            {
                It does not do anything
            }
        }
    }
    \ElseIf{unsettled agent $a_j$ from $G_1$ and $a_j.success=False$}
    {
        \If{unsettled agents $a_j$ from $G_1$ is present at node $v$ with $a_j.state=backtrack$ or $explore$}
        {
            \If{$a_j.count_1+1< 16n^2$}
            {
                It does not do anything.
            }
            \ElseIf{$a_j.count_1+1=16n^2$}
            {
                let $\{a_1, a_2, ..., a_x\}$ be the unsettled agents present at the current node in the increasing order of their IDs. Let $L=a_i.grp\_label=y$, $i\in[1,\,x]$ \\
                \If{$i\leq \lceil{\frac{x}{2}}\rceil$}
                {
                    set $a_i^v(G_1).(parent, label, dfs\_label)=(-1, \,L.0,\, a_j.dfs\_label+1)$
                }  
                \Else
                {
                    set $a_i^v(G_2).(parent, label, dfs\_label)=(-1, \,L.1, \, a_j.dfs\_label+1)$\\
                }   
            }
        }
    }
}
\ElseIf{$r$ mod 2=0}
{
    It does not do anything.
}
\end{algorithm}     
 
     \item[(4)] \textbf{If} $\bm{a_i.state=backtrack}$ \textbf{and} $\bm{a_i.success=False}:$ Two cases are possible (lines 24-40 of Algorithm \ref{a_i.state:backtrack}).
     \vspace{0.15cm}
    \begin{itemize}
        \item[(a)] \textbf{There is an agent from $\bm{G_1}$ with $\bm{settled=0}$ at node $\bm{u}$:}  Let $a_j$ be settled at node $u$, $b_i \in G_1$ be at node $u$, and $b_i.settled=0$. In this case, there are two cases. 
        \begin{itemize}
            \item[(i)]  If $a_i.prt\_out=b_i.prt\_out$ (observe that based on communication, agent $a_i$ can compute what agent $b_i$ computes, and vice versa), then agent $a_i$ does the following. It updates $a_i.count_2=0$ and sets $a_i.skip=a_i.prt\_out$, $a_i.prt\_out=$ minimum available port except $a_i.skip$, \\$a_i.dfs\_label=a_i.dfs\_label+1$. Agent $a_i$ moves through $a_i.prt\_out$.

            \item[(ii)]  If $a_i.prt\_out\neq b_i.prt\_out$, it updates $a_i.count_2=a_i.count_2+1$. If $a_i.count_2<4n^2$, then move through $a_i.prt\_out$. Else, if $a_i.count_2=4n^2$, then $a_i$ updates $a_i.count_2=0$ and $a_i.count_3=a_i.count_3+1$. If $a_i.count_3<4n^2$, then agent $a_i$ sets $a_i.state=explore$, $a_i.port=a_i.prt\_out$, $a_i.dfs\_label=a_i.dfs\_label+1$, $a_i.prt\_out=0$, $a_i.skip=-1$, and moves through $a_i.prt\_out$. Note that $a_i.count_3$ always remains less than $4n^2$ in this case as justified in the correctness.
        \end{itemize}
        
        \item[(b)] \textbf{There is no agent from $\bm{G_1}$ with $\bm{settled=0}$ at node $\bm{u}$:} It updates $a_i.count_2=a_i.count_2+1$. If $a_i.count_2<4n^2$, then move through $a_i.prt\_out$. Else, if $a_i.count_2=4n^2$, then $a_i$ updates $a_i.count_2=0$ and $a_i.count_3=a_i.count_3+1$. If $a_i.count_3<4n^2$, then agent $a_i$ sets $a_i.state=explore$, $a_i.dfs\_label=a_i.dfs\_label+1$, $a_i.skip=a_i.prt\_out$, $a_i.prt\_out=$ minimum available port except $a_i.skip$, and moves through $a_i.prt\_out$. Else, if $a_i.count_3=4n^2$, it divides into two groups using $Group\_divide$ procedure mentioned in the algorithm for unsettled agents of $G_1$.
\end{itemize}
 \end{itemize} 
\begin{algorithm}
    \caption{Algorithm for settled agent $a_i$ at node $v$ and at least one unsettled agent $a_j$ from $G_2$ is present at node $v$}\label{algo:settled-2}
    \If{$r$ mod 2=1}
    {
        \If{unsettled agent $a_j$ from $G_2$ and $a_j.success=True$}
        {
            \If{$a_j.state=backtrack$}
            {
                \If{$parent$ component of $a_i^v(G_2)$ is $-1$}
                {
                    \If{$(a_j.prt\_out+1)$ mod $\delta_v=$ minimum available port except for $a_j.skip$}
                    {
                        set $a_i^v(G_2).(parent,\, label, \, dfs\_label)=(-1, \, a_j.grp\_label,\,  a_j.dfs\_label+1)$
                    }
                    \ElseIf{$(a_j.prt\_out+1)$ mod $\delta_v\neq $ minimum available port except for $a_j.skip$}
                    {
                        \If{$(a_j.prt\_out+1)$ mod $\delta_v = a_j.skip$}
                        {
                            \If{$(a_j.prt\_out+2)=\delta_v-1$}
                            {
                                set $a_i^v(G_2).(parent,\, label, \, dfs\_label)=(-1, \, a_j.grp\_label,\,  a_j.dfs\_label+1)$
                            }
                            \ElseIf{$(a_j.prt\_out+2)\neq \delta_v-1$}
                            {
                                It does not do anything
                            }
                        }
                        \ElseIf{$(a_j.prt\_out+1)$ mod $\delta_v \neq a_j.skip$}
                        {
                            It does not do anything
                        }
                    }
                }
                \ElseIf{$parent$ component of $a_i^v(G_2)$ is not -1}
                {
                    It does not do anything
                }
            }
            \If{$a_j.state=explore$}
            {
                \If{$label$ and $dfs\_label$ components from $a_i^v(G_2)$ do not match with $a_j.grp\_label$ and $a_j.dfs\_label$}
                {
                    set $a_i^v(G_2).(parent,\, label, \, dfs\_label)=(a_j.prt\_in, \, a_j.grp\_label,\,  a_j.dfs\_label)$
                }
                \If{$label$ and $dfs\_label$ components from $a_i^v(G_2)$ match with $a_j.grp\_label$ and $a_j.dfs\_label$}
                {
                    It does not do anything
                }
            }
        }
        \ElseIf{unsettled agent $a_j$ from $G_2$ and $a_j.success=False$}
        {
            \If{unsettled agents $b_i$ from $G_1$ is present at node $v$}
            {
                \If{$b_i.prt\_out = a_j.prt\_out$ and $a_j.state = backtrack$}
                {
                    $a_i^v(G_2).(parent, \,label, \, dfs\_label)=(-1, \, a_i.grp\_label\, a_i.dfs\_label+1)$
                }
                \ElseIf{$b_i.prt\_out \neq a_j.prt\_out$ and $a_j.state = backtrack$ or there is no agent from $G_1$}
                {
                    \If{$a_j.count_2+1= 4n^2$}
                    {
                        \If{$a_j.count_3+1<4n^2$}
                        {
                            set $a_i^v(G_2).(parent, label, dfs\_label)=(-1, a_j.grp\_label, a_j.dfs\_label+1)$
                        }
                        \ElseIf{$a_j.count_3+1=4n^2$}
                        {
                            let $\{a_1, a_2, ..., a_x\}$ be the unsettled agents present at the current node in the increasing order of their IDs. Let $L=a_j.grp\_label$, $j\in[1,\,x]$ \\
                            \If{$j\leq \lceil{\frac{x}{2}}\rceil$}
                            {
                                set $a_i^v(G_1).(parent, label, dfs\_label)=(-1, \,L.0, \,a_j.dfs\_label+1)$
                            }  
                            \Else
                            {
                                set $a_i^v(G_2).(parent, label, dfs\_label)=(-1, \, L.1, \,a_j.dfs\_label+1)$
                            }   
                        }
                    }
                    \ElseIf{$a_j.count_2+1< 4n^2$}
                    {
                        It does not do anything
                    }
                }
                \ElseIf{$b_i.prt\_out = a_j.prt\_out$ and $a_j.state = explore$}
                {
                    \If{$parent$ component of $a_i^v(G_2)$ is $-1$}
                    {
                        \If{$a_i.prt\_out+1=\delta_v -1$}
                        {
                            $a_i^v(G_2).(parent, label, dfs\_label)=(-1, \, a_i.grp\_label, \,a_j.dfs\_label+1)$
                        }
                        \ElseIf{$a_i.prt\_out+1 \neq \delta_v -1$}
                        {
                            It does not do anything
                        }
                    }
                }
                \ElseIf{$b_i.prt\_out \neq a_j.prt\_out$ and $a_j.state = explore$ or there is no agent from $G_1$}
                {
                    \If{$a_j.count_2+1= 4n^2$}
                    {
                        \If{$a_j.count_3+1<4n^2$}
                        {
                            set $a_i^v(G_2).(parent, label, dfs\_label)=(-1, a_j.grp\_label, a_j.dfs\_label+1)$
                        }
                        \Else
                        {
                            let $\{a_1, a_2, ..., a_x\}$ be the unsettled agents present at the current node in the increasing order of their IDs. Let $L=a_j.grp\_label$, $j\in[1,\,x]$ \\
                            \If{$j\leq \lceil{\frac{x}{2}}\rceil$}
                            {
                                set $a_i^v(G_1).(parent, label, dfs\_label)=(-1, \,L.0, \, a_j.dfs\_label+1)$
                            }  
                            \Else
                            {
                                set $a_i^v(G_2).(parent, label, dfs\_label)=(-1, \,L.1, \,a_j.dfs\_label+1)$\\
                            }   
                        }
                    }
                    \ElseIf{$a_j.count_2+1< 4n^2$}
                    {
                        It does not do anything
                    }
                }
            }
        }
    }
    \ElseIf{$r$ mod 2=0}
    {
        It does not do anything. 
    }
\end{algorithm}

\noindent \textbf{Algorithm for unsettled agent $\bm{a_i}$ at even round:}
At even round, each unsettled agent $a_i$ of $G_1$ or $G_2$ updates $a_i.success$ parameter. If their movement in the previous rounds is successful, agent $a_i$ updates their $a_i.success=True$. Otherwise, it updates their parameter $a_i.success=False$. If $a_i.state=explore$, then refer to lines 41-45 of Algorithm \ref{a_i.state:backtrack}. Otherwise, refer to lines 55-59 of Algorithm \ref{a_i.state:explore}.

\noindent \textbf{Algorithm for settled agents}: (refer to Algorithm \ref{algo:settled-1} and Algorithm \ref{algo:settled-2}) Let $a_i$ be at node $u$ with $a_i.settled=1$. It has two parameter $a_i^u(G_1)$ and $a_i^u(G_2)$. At each odd round, based on the communication, it computes the decision of $G_1$ and $G_2$, and it stores $G_1$ information in $a_i^u(G_1)$ and $G_2$ information in $a_i^u(G_2)$. In the move, it does not do anything. At even rounds, it does not do anything.

\subsubsection{Correctness and analysis of algorithm}
Let $m$ be the number of edges in $G$. At the beginning of the algorithm, $n+1$ agents are co-located at a node of $G$, and agents are aware of parameter $n$. Initially, each agent $a_i$ is part of $G_1$ as $a_i.grp\_label=10$. Assume that the adversary does not remove an edge in the first $4m$ odd rounds. In this case, all agents reach the dispersion configuration because this is nothing but the dispersion algorithm using the DFS traversal of \cite{Augustine_2018}. Hence, all agents reach dispersion configuration in the first $4m$ odd rounds. After including even rounds, agents achieve dispersion in the first $8m$ rounds.

At round $r<8m$, all agents in $G_1$ are at node $u$ and want to go through edge $e= \overline{uv}$. At round $r$, if the adversary removes an edge $e$, then the movement of agents in $G_1$ becomes unsuccessful. In this case, the group $G_1$ is divided into two new groups $G_1$ and $G_2$. Both groups consider the current node $u$ as the root node for the new DFS traversal and restart the algorithm for DFS via port 0. Let $r'\geq r$ be round when agents in $G_1$ are at node $u'$ and want to go through edge $e'= \overline{u'v'}$ via port $p$, and their movement becomes unsuccessful due to the absence of the edge $e'$ at round $r'$. We have the following claim.

\begin{claim}\label{claim:1}
    If edge $e'$ does not appear within the next $16n^2$ odd rounds after it is deleted by the adversary in round $r'$, then agents in $G_2$ get dispersed in those $16n^2$ rounds.
\end{claim}
\begin{proof}
   Suppose at round \( r' \), the agents in \( G_2 \) are at node \( u'' \). Within the next \( 8m \) odd rounds, one of two things is possible: the agents in \( G_2 \) achieve dispersion, or the unsettled agents from \( G_2 \) attempt to move via the edge \(\overline{v'u'}\). This is due to the fact that during the first \( 4m \) odd rounds, agents in \( G_2 \) either achieve dispersion, start a new DFS from some node \( w \), or reach the root, say \( v_r \), of the current DFS traversal of agents in \( G_2 \). From node \( v_r \) (or respectively \( w \)), agents in \( G_2 \) start the new DFS by incrementing their \( dfs\_label \). In the next \( 4m \) odd rounds, the unsettled agents from \( G_2 \) reach node \( v' \) and attempt to move via the edge \( \overline{v'u'} \) due to the following reason: if an agent reaches node \( u' \), then it reaches in $explore$ state, and it moves through the next available port.

After reaching node \( v' \), it tries to move via edge \(\overline{v'u'}\). Since edge \( e' \) is not present, its \( count_2 \) parameter reaches \( 4n^2 \). After this, agents in \( G_2 \) store the port corresponding to the edge \(\overline{v'u'}\) in their \( skip \) parameter. Then, during the next \( 4m \) odd rounds, unsettled agents in \( G_2 \) visit each node of the graph \( G - \{e'\} \) as they reach node \( u' \) to move via edge \(\overline{u'v'}\) in $explore$ state, and move via the next available port, and whenever they reach node \( v' \), they pick a port to explore, excluding the port in their \( skip \) parameter. The total time from round \( r' \) for agents in \( G_2 \) to achieve dispersion is \( 12m + 4n^2 \leq 16n^2 \) odd rounds. This completes the proof.
\end{proof}

If the edge $e'$ reappears within $16n^2$ odd rounds, then agents in $G_1$ move at least one round of their current DFS. Now we make a claim that guarantees the dispersion of agents in $G_1$.

\begin{claim}\label{claim:2}
    For any unsettled agent $a_i\in G_2$, if $a_i.count_3$ matches with $4n^2$, then agents in $G_1$ are already dispersed.
\end{claim}

\begin{proof}
    Let $r''$ be an odd round when edge $e'$ reappears and agents in $G_1$ move at least one round on its current DFS traversal. Note that $r'' >r'$. After round $r''$ onwards, the following cases are possible. 

    \begin{itemize}
        \item \textbf{Case (A) Agents in group $G_1$ and $G_2$ want to go through the same deleted edge:} At round $r'''\geq r''$, unsettled agents in $G_1$ and $G_2$ want to go through the same edge $e_1= \overline{u_1v_1}$. As discussed in Claim \ref{claim:1}, if edge $e_1$ does not appear within the next $12n^2$ odd rounds, then $G_2$ gets dispersed. Otherwise, unsettled agents in $G_1$ move on its current DFS traversal path. 
       
        \item \textbf{Case (B) Agents in $G_2$ want to go through deleted edge:} In this case, agents in $G_2$ at node $u_2$ try to move via edge $e_2=(u_2,\,v_2)$ using port $p_2$. Agent $a_i \in G_2$ has parameter $a_i.count_2$ and it increases this parameter by 1 whenever its movement becomes unsuccessful. If its movement is successful, then $a_i.count_2=0$. If at some odd round, say $t'$, $a_i.count_2$ reaches to $4n^2$ (it means $4n^2$ consecutive odd rounds edge $e'$ is missing), then agent $a_i \in G_2$ restarts the new DFS and stores $a_i.skip=p_2$ corresponding edge $e'$ through which it is trying to move. It is important to note that if edge $e_2$ is missing for consecutive $4n^2$ odd rounds, then two cases are possible, which are as follows. 
        
        \vspace{0.2cm}
        \begin{itemize}
            \item \textbf{Case (a)}: While $G_2$ is waiting for $e_2$ to appear, then within $4n^2$ odd rounds unsettled agents in $G_1$ get stuck at node $u_2$ (or $v_2$). If each unsettled agent $a_j \in G_1$ reaches node $u_2$ and wants to go through edge $e_2$, then it is similar to Case (A). If each unsettled agent $a_j \in G_1$ gets stuck at node $v_2$, then it tries to move through edge $\overline{v_2u_2}$ and starts increasing $a_j.count_1$ parameter by 1 in each unsuccessful movement. When $a_i.count_2=4n^2$, then agent $a_i \in G_2$ restarts new DFS, stores $a_i.skip=p_2$ and increases parameter $a_i.count_3$ by one. If edge $e_2$ does not appear next odd $4m$ rounds (note that $a_j.count_1$ for agent $a_j \in G_1$ is less than $16n^2$), then agents in $G_2$ achieve the dispersion. It is because even if it reaches node $v_2$ and wants to go through edge $\overline{v_2u_2}$ via some port $p$, it reaches as $a_i.state=explore$. Agent $a_i$ changes the outgoing port by $(p+1)$ mode $\delta_{v_2}$ and continues its DFS traversal on $G-\{e_2\}$. Therefore, if edge $e_2$ does not appear in the next consecutive $4m$ odd rounds, then the group $G_2$ is successfully dispersed at the nodes of $G$. If edge $e_2$ reappears, then agents in $G_1$ move at least one round of their current DFS.
            \vspace{0.2cm}
            \item \textbf{Case (b)}: The DFS traversal of $G_1$ is unaffected due to the missing edge $e_2$. In this case, agents in $G_1$ get dispersed. 
        \end{itemize} 
    \end{itemize}

Therefore, if $a_i.count_3$ reaches $4n^2$, and agents in $G_2$ are not dispersed, then they can understand that agents in $G_1$ are dispersed, as whenever the unsettled agent $a_i \in G_2$ increases the parameter $a_i.count_3$, then within the next $4m$ odd rounds, either agent $a_i \in G_2$ achieves dispersion, or agents $a_j \in G_1$ move at least one round of their current DFS. Therefore, unsettled agents in $G_1$ already have executed their dispersion algorithm via DFS for at least $4 \cdot 4m$ odd rounds. This completes the proof.
\end{proof}

Based on Claim \ref{claim:1} $\text{and}$ \ref{claim:2}, we have the following lemma. 
\begin{lemma}\label{lm:correct}
The algorithm correctly solves the $k$-balanced dispersion problem in 1-bounded 1-interval connected graphs for $k = n+1$ co-located agents.
\end{lemma}

\begin{proof}
    Due to Claim \ref{claim:1} $\text{and}$ 2, we can say either $G_1$ or $G_2$ is dispersed and unsettled agents in $G_1\,(G_2)$ understand whether agents in $G_2\,(G_1)$ are dispersed. And, agents in $G_1\,($or $G_2)$ divide into two new groups. In this way, in the end, group $G_1$ (and $G_2$) contains 1 agent. Both groups start running our algorithm, and due to Claim \ref{claim:1} $\text{and}$ \ref{claim:2}, we can say that either agent in $G_1$ or $G_2$ is settled at some node. Therefore, our algorithm correctly solves the dispersion problem in 1-bounded 1-interval graphs when $n+1$ agents are co-located at some node initially. This completes the proof.
\end{proof}

 Due to Lemma \ref{lm:correct}, we can say that one group is divided into two groups only after the other group achieves dispersion. Since we have $n+1$ agents, in this way, all agents are getting dispersed except one agent, which will be part of the last remaining group.

\vspace{0.15cm}

\begin{lemma}\label{lm:time}
    The time complexity of our algorithm is $O(n^4 \cdot \log n)$.
\end{lemma}
\begin{proof}
    In Lemma \ref{lm:correct}, we have shown that the agents achieve the dispersion. If the adversary does not remove any edge in the first $4m$ odd rounds then all agents are dispersed except one agent. In an odd round $r$ (where $r<8m$), if the adversary removes the edge $e$ and the group of agents wants to go through the edge $e$ according to their DFS algorithm, then they divide into two groups. We show that in the next $64n^4$ odd rounds, either $G_1$ or $G_2$ gets dispersed. We have proved that whenever $a_i \in G_2$ increases $a_i.count_3$ by 1, then within the next $4m$ odd rounds, either agent $a_i \in G_2$ achieves dispersion, or agents in $a_j \in G_1$ move at least one round of their current DFS. In Claim \ref{claim:1}, we have shown that if agents in $G_1$ get stuck for $16n^2$ odd rounds for some edge, then $G_2$ is dispersed. In Claim \ref{claim:2}, we have shown that whenever the unsettled agent $a_i \in G_2$ increases the parameter $a_i.count_3$, then within the next $4m$ odd rounds, either agent $a_i \in G_2$ achieves dispersion, or agents $a_j \in G_1$ move at least one round of their current DFS. It implies that if agents in $G_2$ are not dispersed and for $a_i \in G_2$, $a_i.count_3$ reaches $4n^2$, then agents in $G_2$ can understand that agents in $G_1$ are dispersed. Hence, in each $16n^2$ consecutive rounds, one of the possibilities occurs: either agents $a_j \in G_1$ move at least one round of their current DFS or $G_2$ gets dispersed. Therefore, between rounds $r$ and $r + 16n^2 \cdot 4n^2$, either the unsettled agents in $G_1$ or the unsettled agents in $G_2$ get dispersed. Including even rounds, we can say that after round $r$, in the next $128n^4$ rounds, either $G_1$ or $G_2$ is dispersed.

    Using Lemma~\ref{lm:correct}, if $G_1$ (or $G_2$) understands that $G_2$ (or $G_1$) is dispersed, then $G_1$ (or $G_2$) divides into two new groups, and both groups contain half the agents of $G_1$ (or $G_2$). These two groups repeat the same procedure. The number of such group divisions is $O(\log n)$, and one of the two groups is dispersed in at most $128n^4$ rounds. Therefore, our algorithm takes $r + 128n^4 \cdot \log n \leq 8m + 128n^4 \cdot \log n = O(n^4 \cdot \log n)$, as $m \leq n^2$, rounds to solve the dispersion problem in 1-bounded 1-interval graphs when $n+1$ agents are co-located at some node initially. This completes the proof.
\end{proof}

We are now presenting the results that show the memory requirements per agent in our algorithm.

\begin{lemma}\label{lm:memory}
    In our algorithm, agents require $\Theta(\log n)$ memory.
\end{lemma}
\begin{proof}
        For agent $a_i$, $a_i.grp\_label$ is nothing but a binary string. Whenever we divide the group, the $a_i.grp\_label$ is appended by 0 or 1. The number of group divisions is $O(\log n)$. Therefore, the length of $a_i.grp\_label$ is $O(\log n)$, which agent $a_i$ can store in its $O(\log n)$ memory. According to Lemma \ref{lm:time}, all agents achieve dispersion in $O(n^4 \cdot \log n)$ rounds. Therefore, $a_i.dfs\_label$ can not be more than $O(n^4 \cdot \log n)$, which agent $a_i$ can store in its $O(\log n)$ memory. Apart from these two parameters, all unsettled agents $a_i$ either store boolean variables or port information. Such parameters are constant as per our algorithm. To store each parameter, we need $O(\log n)$ storage. Therefore, each unsettled agent $a_i$ needs $O(\log n)$ memory. All settled agents maintain two parameters as per our algorithm $a_i^u(G_1)$ and $a_i^u(G_2)$. To store each parameter, settled agents need $O(\log n)$ memory. Therefore, each settled agent $a_i$ needs $O(\log n)$ memory. 
    
    Due to Lemma \ref{lm:memory_lower}, the memory lower bound for the dispersion of $k\geq n+1$ agents is $\Omega(\log n)$, and in our algorithm, agents use $O(\log n)$ memory each. Therefore, the agents in our algorithm require $\Theta(\log n)$ memory due to $k=n+1$. 
\end{proof}

The following theorem is an outcome of Lemma \ref{lm:correct}, \ref{lm:time}, and \ref{lm:memory}.

\begin{theorem}
Our algorithm solves $k$-balanced dispersion in 1-bounded 1-interval connected graphs with $k = n+1$ co-located agents in $\mathcal{G}$ within $O(n^4 \log n)$ rounds, using $\Theta(\log n)$ memory per agent.
\end{theorem}

\begin{observation}\label{obs:termination}
   All agents know the value of $n$ and $m \leq n^2$. If all agents want to achieve the termination, they can terminate after $8n^2+128n^4 \cdot \lceil \log n \rceil$ rounds. 
\end{observation}

\subsection{Dispersion with $k = pn + 1 (\text{or }pn+2)$ co-located agents}
In this section, we extend the algorithm from Section \ref{sec:gen_graph} to solve dispersion with $k$ co-located agents in 1-bounded 1-interval connected graphs, where $k=pn+q, q\in \{1,2\}, \;p\in \mathbb{N}\cup \{0\}$. We refer to the algorithm described in Section~\ref{sec:algorithm} as $Algo\_Int\_Disp()$. We assume that agents are aware of the parameter $n$. Here, agents have two additional parameters.

\begin{itemize}
    \item $a_i.r:$ Agent $a_i$ uses this parameter to count rounds since the start of the algorithm. Initially, $a_i.r=0$, and in each round, agent $a_i$ increases it by 1.
    \item $a_i.ID_1:$ This parameter is used by agents to store the new assigned ID. Initially, it is $a_i.ID_1=\bot$.
\end{itemize}

We divide our algorithm into three phases, and the details of each phase are as follows.

\begin{itemize}
    \item \textbf{Phase 1:}  
    At round $r=0$, since agents know $n$ and are co-located at node $v$, they compute the value of $p$ and $q$.

    At round 1, if $p=0$, then they do the following based on $deg(v)$.
    
\begin{itemize}
  \item If \( \deg(v) = 1 \), one agent remains at \( v \), and the other moves through port 0; both terminate.
  \item If \( \deg(v) \geq 2 \), then one agent moves via port 0 and the other via port 1; both terminate in the next round.
\end{itemize}

Else if $p\geq 1$, at round $r=1$, agents distribute themselves into $n$ groups: $g_1, g_2, \ldots, g_n$. Each group $g_i$, for $1 \leq i \leq n-1$, contains $p$ agents, while $g_n$ contains $p + q$ agents. Every agent $a_i$ in group $g_j$ stores $a_i.ID_1=j$. After this grouping, they proceed to the next phase.

    \item \textbf{Phase 2:}  
    From the second round onward, agents begin executing $Algo\_Int\_Disp()$ with the following modification: agents in group $g_i$ behave as a single agent as ID $i$. Just as in earlier algorithms, when a hole is found, one agent gets settled. Here, when a group $g_i$ finds a hole, all agents in $g_i$ settle together. Similarly, when a group divides into two groups, we divide the list $g_1, g_2, \ldots, g_l$ into two parts: $g_1, \ldots, g_{\lfloor l/2 \rfloor}$ becomes $G_1$, and $g_{\lfloor l/2 \rfloor+1}, \ldots, g_l$ becomes $G_2$.\footnote{After forming these two groups, group $G_1$ contains $\lfloor l/2 \rfloor$ $g_i$s, and group $G_2$ contains $\ell - \lfloor l/2 \rfloor$ $g_i$s.} These $g_i$s stop executing $Algo\_Int\_Disp()$ once group $G_1$ (respectively, $G_2$) at round $r$ determines that group $G_2$ (respectively, $G_1$) is dispersed, and that the size of $G_1$ (respectively, $G_2$) is 1. Let $g_i$ be a member of $G_1$ (respectively, $G_2$) and be at node $u$ at the end of round $r$. At round $r$, $G_1$ (or $G_2$) proceeds to the next phase.

   \item \textbf{Phase 3:} Since the agents are divided into $n$ groups, it is possible that at the end of Phase 2, node $u$ has another group $g_j$ that is already settled. In such a case, there exists a node, say $v_h$, that remains a hole at the beginning of round $r+1$. At most $p+1$ agents and at least $p$ agents need to settle at node $v_h$. Note that it is possible that $u = v_h$ at round $r+1$. We now describe how the agents fill the hole at $v_h$ using $Algo\_Int\_Disp()$. There are two cases to consider:
\begin{itemize}
    \item \textbf{Case 1 $(i=n)$:} In this case, $g_n$ contains $p+q$ agents. At round $r+1$, the agents in group $g_n$ begin executing $Algo\_Int\_Disp()$ as $p+q$ co-located agents at node $u$, in place of $n+1$ co-located agents (i.e., $G_1$ contains $p+q$ agents). While execution of $Algo\_Int\_Disp()$, whenever \( G_1 \) (or \( G_2 \)) reaches node, say \( v_h \), that has fewer than \( p \) agents, it settles agents there such that the number of agents at \( v_h \) does not exceed \( p+1 \).

    If \( G_1 \) (respectively \( G_2 \)) detects that \( G_2 \) (respectively \( G_1 \)) is fully dispersed and its own size is 1, it does the following. Without loss of generality, assume that the size of $G_1$ is 1, and agent $a_i\in G_1$ is at node $v$. In this case, if node $v$ has $p$ agents except agent $a_i$, then agent $a_i$ terminates. Else if node $v$ has $p+1$ agents except agent $a_i$, then the minimum ID agent $a_j$ (except agent $a_i$) does the following.

    \begin{itemize}
  \item If \( \deg(v) = 1 \), agent $a_j$ remains at \( v \), and agent $a_i$ moves through port 0; both terminate.
  \item If \( \deg(v) \geq 2 \), then agent $a_i$ moves via port 0 and agent $a_j$ via port 1; both terminate in the next round.
\end{itemize}

    Otherwise, if $G_1$ (respectively, $G_2$) detects that $G_2$ (respectively, $G_1$) is fully dispersed and its own size is 2, then it proceeds as follows. In the analysis, we show that at this point, node $v_h$ contains $p$ agents. Without loss of generality, assume that the size of $G_1$ is 2, and both agents are at node $v$.

\begin{itemize}
  \item If \( \deg(v) = 1 \), one agent remains at \( v \), and the other moves through port 0; both terminate.
  \item If \( \deg(v) \geq 2 \), then one agent moves via port 0 and the other via port 1; both terminate in the next round.
\end{itemize}

\item \textbf{Case 2 $(i<n)$:} By Observation 1, $n+1$ agents achieve dispersion in at most $8n^2+128n^4 \cdot \lceil \log n \rceil$ rounds. Let $T$ denotes $8n^2+128n^4 \cdot \lceil \log n \rceil$. Therefore, Phase 2 ends within the first $T+2$ rounds. In this case, group \( g_n \), consisting of \( p+q \) agents, had already settled at some node \( u_1 \) in a previous round. Without loss of generality, let $q=2$. Let \( b_1 \) and \( b_2 \) be the two agents with the smallest IDs at node \( u_1 \).

After \( b_1.r = b_2.r = T + 2 \) rounds, agents $b_1$ and $b_2$ form \( G_2 \). Similarly, after After \( a_j.r =T + 2 \) for every $a_j\in g_i$ becomes part of group \( G_1 \). From round $T+3$, agents start executing $Algo\_Int\_Disp()$ in the following way.

While execution of $Algo\_Int\_Disp()$, whenever \( G_1 \) (or \( G_2 \)) reaches node, say \( v_h \), that has fewer than \( p \) agents, it settles agents there such that the number of agents at \( v_h \) does not exceed \( p+1 \). 
    
If \( G_1 \) (respectively \( G_2 \)) detects that \( G_2 \) (respectively \( G_1 \)) is fully dispersed and its own size is 1, it does the following. Without loss of generality, assume that the size of $G_1$ is 1, and agent $a_i\in G_1$ is at node $v$. In this case, if node $v$ has $p$ agents except agent $a_i$, then agent $a_i$ terminates. Else if node $v$ has $p+1$ agents except agent $a_i$, then the minimum ID agent $a_j$ (except agent $a_i$) does the following.

    \begin{itemize}
  \item If \( \deg(v) = 1 \), agent $a_j$ remains at \( v \), and agent $a_i$ moves through port 0; both terminate.
  \item If \( \deg(v) \geq 2 \), then agent $a_i$ moves via port 0 and agent $a_j$ via port 1; both terminate in the next round.
\end{itemize}

Otherwise, if $G_1$ (respectively, $G_2$) detects that $G_2$ (respectively, $G_1$) is fully dispersed and its own size is 2, then it proceeds as follows. In the analysis, we show that at this point, node $v_h$ contains $p$ agents. Without loss of generality, assume that the size of $G_1$ is 2, and both agents are at node $v$.

\begin{itemize}
  \item If \( \deg(v) = 1 \), one agent remains at \( v \), and the other moves through port 0; both terminate.
  \item If \( \deg(v) \geq 2 \), then one agent moves via port 0 and the other via port 1; both terminate in the next round.
\end{itemize}
\end{itemize}
\end{itemize}

% According to Claim 1, if \( G_1 \) waits for \( 16n^2 \) consecutive odd rounds, then \( G_2 \) will have visited every node at least once, guaranteeing a meeting between \( G_1 \) and \( G_2 \). The two groups then combine to form a group of size \( p+q \), reducing the scenario to Case 1.

% If \( G_1 \) does not wait for \( 16n^2 \) odd rounds, it continues the DFS traversal. In at most \( 16n^2 \times 4n^2 = 64n^4 \) odd rounds, it reaches some node \( u \) and settles there. If \( G_2 \) does not encounter \( G_1 \) during this time, it concludes that all agents in \( G_1 \) are already settled. Then \( G_2 \) performs the following termination procedure.

% Suppose the $q$ agents in \( G_2 \) are at node \( v \). If $q=1$, it terminates; otherwise (i.e., $q=2$), it does the following.

% \begin{itemize}
%   \item If \( \deg(v) = 1 \), one agent remains, and the other moves through port 0; both terminate.
%   \item If \( \deg(v) \geq 2 \), then one agent moves via port 0, and the other via port 1; both terminate in the next round.
% \end{itemize}
% \end{itemize}
% \end{itemize}

Based on this, we have the final theorem.
\begin{theorem}
Our algorithm solves $k$-balanced dispersion in 1-bounded 1-interval connected graphs with $k = pn+q, \text{ where }, p\in \mathbb{N}\cup 0, q\in \{1,2\}$, co-located agents in $\mathcal{G}$ within $O( n^4 \cdot max(\log n,\, \log p))$ rounds, using $O\big(\max\{\log n, \log p\}\big)$ memory per agent.
\end{theorem}
\begin{proof}
Since the agents know the value of $n$, they can compute $p$ and $q$ correctly in Phase 1. At the end of this phase, we have $n$ groups $g_i$ of agents. In Phase~2, these $g_i$s form $G_1$ and start executing the algorithm described in Section \ref{sec:algorithm}. This phase ends when $G_1$ (respectively, $G_2$) determines that $G_2$ (respectively, $G_1$) is dispersed and the size of set $G_1$ is 1 (respectively, the size of set $G_1$ is 1). This is guaranteed by Lemma~\ref{lm:correct}, and the phase completes in $O(n^4 \cdot \log n)$ rounds by Lemma~\ref{lm:time}. At this point, one group $g_i$ may remain unsettled, and there exists a hole at node $v_h$.

In Phase 3, we consider two cases for the group $g_i$.  
In the first case (i.e., $i=n$), the $p+q$ agents (with $q \in \{1,2\}$) are co-located at node $w$. The agents in $g_i$ becomes $G_1$, and execute $Algo\_Int\_Disp()$. As per Lemma \ref{lm:time}, in $8m+128n^4$ rounds, either $G_1$ or $G_2$ archives dispersion using $Algo\_Int\_Disp()$. It implies that either $G_1$ or $G_2$ is able to visit every node of $G$ at least once. Therefore, either $G_1$ or $G_2$ reaches node $v_h$. In this way, they can get settled at node $v_h$. Since $p+q$ agents are co-located, therefore using Lemma \ref{lm:time}, node $v_h$ has at least $p$ agents in $O(n^4\cdot \log(p+q))=O(n^4\cdot \log p)$. After $O(n^4 \cdot \log p)$ rounds, $G_1$ (respectively, $G_2$) determines that $G_2$ (respectively, $G_1$) is dispersed, and the size of $G_1$ (respectively, $G_2$) is 1 or 2, due to Lemma~\ref{lm:correct}. Now, suppose at most two agents remain unsettled. If both are not at the same node, then they can terminate. Otherwise, if both are at the same node $v$, then they can understand this as the number of agents at node $v$ will be $p+2$. The adversary cannot remove an edge incident to a node of degree 1. If the $q$ agents are at node $v$ with $\deg(v) = 1$, then one can settle at $v$ and the other at its neighbour. If $\deg(v) \geq 2$, then one agent tries to move via port 0 and the other via port 1. Since the adversary can remove at most one edge, and $G$ is simple, at least one of the two reaches a new node. The other either stays at $v$ (due to a missing edge) or also reaches a new node. In both situations, dispersion is completed.

In the second case (i.e., $i\neq n$), agents in $g_i$ are at node $w$. In this case, $g_n$ has already settled. As per our algorithm, after $T+2$ rounds since beginning from $g_n$, if $q=1$, then the minimum ID agents from $g_n$ become part of $G_2$; else if $q=2$, then the first two minimum ID agents become part of $G_2$. All agents in $g_i$ become part of $G_1$. From $T+3$ rounds, both groups start executing the algorithm from Section~\ref{sec:algorithm}. The correctness of this case is similar to the first case. Therefore, in Phase 3, the $p+q$ agents get settled using $Algo\_Int\_Disp()$, and require $O(n^4 \cdot \log p)$ rounds, as the $p+q$ agents divide into two groups recursively until each group has size at most 2 (similar to the $n+1$ case discussed in Lemma~\ref{lm:time}). Hence, the total time complexity is $O(n^4 \cdot \log n) + O(n^4 \cdot \log p) = O(n^4 \cdot \max\{\log n, \log p\}).$

Since each agent $a_i$ needs to remember $n$, $p$, its own round count $a_i.r = O(n^4 \cdot \max\{\log n, \log p\})$, and its new ID $a_i.ID_1$, $1 \leq a_i.\mathit{ID}_1 \leq n$, and due to Lemma~\ref{lm:memory}, the memory requirement is $O\big(\max\{\log n, \log p\}\big)$ per agent. This completes the proof.
\end{proof}

\begin{remark}
    If $1 \leq k \leq n^{c_1}$ for some constant $c_1$, then by Lemma~\ref{lm:memory_lower}, the memory required by the agents is $O\big(\max\{\log n, \log p\}\big) = \Theta(\log n)$,
    since $p \leq k \leq n^{c_1}$. Hence, in this range of $k$, the memory requirement by the agents is optimal.
\end{remark}

\section{Conclusion}
In this work, we established a clear connection between dispersion and load balancing through the notion of $k$-balanced dispersion. We analyzed the problem under different connectivity models and showed that while prior works \cite{Ajay_dynamicdisp,Saxena_2025_} leave no gap under their restricted communication assumptions, our setting reveals a nontrivial gap between necessary and sufficient conditions. An interesting future direction is to determine for which values of $l$ the problem becomes unsolvable when agents are allowed $l$-hop communication.  

We also designed an algorithm that solves $k$-balanced dispersion in 1-bounded 1-interval connected graphs for $k = pn + q$ agents with $q \in \{1,2\}$ under face-to-face communication and zero visibility. Extending this result to the case $q \in [3, n-1]$ appears technically challenging if agents are given f-2-f communication and 0-hop visibility, and remains a compelling open problem.

\section{Acknowledgement}
Ashish Saxena would like to acknowledge the financial support from IIT Ropar. Kaushik Mondal would like to acknowledge the ISIRD grant provided by IIT Ropar. This work was partially supported by the FIST program of the Department of Science and Technology, Government of India, Reference No. SR/FST/MS-I/2018/22(C).

\vspace{0.4cm}
\noindent\textbf{Declaration of generative AI and AI-assisted technologies in the writing process} 

During the preparation of this work, we used \emph{Grammarly, QuiltBot and InstaText} tools in order to improve language quality. After using this tool, we reviewed and edited the content as needed and take full responsibility for the content of the publication.

\bibliography{bib}
\newpage

\end{document}